\documentclass[a4paper,UKenglish,cleveref, autoref, nameinlink,capitalize, thm-restate, 10pt]{article}
\usepackage{microtype}
\usepackage{complexity}
\usepackage[margin=3cm]{geometry}
\usepackage{amsmath,dsfont}
\usepackage{mathtools}
\usepackage[color=yellow]{todonotes}
\usepackage{tikz}
\usetikzlibrary{decorations.pathmorphing}
\usetikzlibrary{automata,positioning}
\usepackage{tcolorbox}
\usepackage{color}
\usepackage{xspace}
\usepackage{paralist}
\usepackage{authblk}
\usepackage{hyperref}
\usepackage{cleveref}
\usepackage{subcaption}
\usepackage{amsthm}
\usepackage{thm-restate}

\hypersetup{
	colorlinks=true,
	linkcolor=blue,
	filecolor=magenta,      
	urlcolor=gray,
	citecolor=magenta
}

\newif\ifarxive
\arxivetrue 
\arxivefalse

\ifarxive

\else

\fi

\graphicspath{{figures/}}

\newtheorem{theorem}{Theorem}
\newtheorem{lemma}[theorem]{Lemma}
\newtheorem{corollary}[theorem]{Corollary}

\Crefname{observation}{Observation}{Observations}
\Crefname{algorithm}{Algorithm}{Algorithms}
\Crefname{section}{Sect.}{Sects.}
\Crefname{observation}{Observation}{Observations}
\Crefname{lemma}{Lemma}{Lemmas}
\Crefname{figure}{Fig.}{Figs.}
\Crefname{corollary}{Corollary}{Corollaries}
\definecolor{realblue}{rgb}{0,0,1}
\definecolor{darkerblue}{rgb}{0.094,0.455,0.804}
\definecolor{darkblue}{rgb}{0.063,0.306,0.545}
\definecolor{red}{rgb}{0.627,0.117,0.156}
\definecolor{green}{rgb}{0,0.588,0.509}
\definecolor{orange}{rgb}{0.903,0.739,0.382}
\definecolor{realred}{rgb}{1,0,0}

\newcommand{\blue}[1]{{{\textcolor{blue}{#1}\xspace}}}

\renewcommand{\emph}[1]{\blue{\em {#1}}}

\newcommand{\st}{st}

\renewcommand{\subparagraph}[1]{\smallskip\noindent{\bf #1}\xspace}

\title{On Upward-Planar L-Drawings of Graphs\thanks{An extended abstract of this work appeared at MFCS 2022. Cornelsen was supported by the DFG – Project-ID 50974019 – TRR 161 (B06). Da Lozzo was supported in part by MIUR grant 20174LF3T8 {\em ``AHeAD''}.}}

\author[1]{Patrizio~Angelini}

\author[2]{Steven Chaplick}

\author[3]{Sabine Cornelsen}

\author[4]{Giordano {Da Lozzo}}

\affil[1]{John Cabot University, Rome, Italy. \texttt{pangelini@johncabot.edu}}

\affil[2]{Maastricht University, The Netherlands. \texttt{s.chaplick@maastrichtuniversity.nl}}

\affil[3]{University of Konstanz, Germany. \texttt{sabine.cornelsen@uni-konstanz.de}}

\affil[4]{Roma Tre University, Rome, Italy. \texttt{giordano.dalozzo@uniroma3.it}}

\begin{document}

	\maketitle

\begin{abstract}
	In an {\em upward-planar L-drawing} of a directed acyclic graph (DAG) each~edge~$e$ is represented as a polyline composed of a vertical segment with its lowest endpoint at the tail of $e$ and of a horizontal segment ending at the head of $e$. Distinct edges may overlap,
	but not cross. 
	Recently, upward-planar L-drawings have been studied for $st$-graphs, i.e., planar DAGs with a single source $s$ and a single sink $t$ containing an edge directed from $s$ to $t$.
	It is known that a {\em plane \st-graph}, i.e., an embedded  $st$-graph in which the edge $(s,t)$ is incident to the outer face,  admits an upward-planar L-drawing if and only if it admits a bitonic \st-ordering, which can be tested in linear time. 
	
	We study upward-planar L-drawings of DAGs that are not necessarily \st-graphs.
	On the combinatorial side, we
	show that a plane DAG admits an upward-planar L-drawing if and only if it is a subgraph of a plane \st-graph admitting a bitonic \st-ordering. 
	This allows us to show that not every tree with a fixed bimodal embedding admits an upward-planar L-drawing.  Moreover, we prove that any acyclic cactus with a single source (or a single sink) admits an upward-planar L-drawing, which respects a given outerplanar~embedding if there are no transitive edges.
	On the algorithmic side,  we consider DAGs with a single source (or a single sink).
	We give linear-time testing algorithms for these DAGs in two cases: (i) when the drawing must respect a prescribed embedding and (ii) when no restriction is given on the embedding, but it is biconnected and series-parallel.
\end{abstract}

\section{Introduction}

In order to visualize hierarchies, directed acyclic graphs (DAGs) are often drawn in such such a way that the geometric representation of the edges reflects their direction. To this aim \emph{upward drawings} have been introduced, i.e., drawings in which edges are monotonically increasing curves in the $y$-direction. Sugiyama et al. \cite{sugiyama_framework} provided a general framework for drawing DAGs upward. To support readability, it is desirable to avoid edge crossings~\cite{DBLP:journals/vlc/Purchase02,DBLP:journals/ivs/WarePCM02}.  However, not every planar DAG admits an \emph{upward-planar drawing}, i.e., an upward drawing in which no two edges intersect except in common endpoints. A necessary condition is that the corresponding embedding is \emph{bimodal}, i.e., all incoming edges are consecutive in the cyclic sequence of edges around any vertex.
Di Battista and Tamassia~\cite{DT-aprad-88} showed that a DAG is upward-planar if and only if it is a subgraph of a \emph{planar \st-graph}, i.e., a planar DAG with a single source and a single sink that are connected by an edge.  Based on this characterization, it can be decided in near-linear 
time whether a DAG admits an upward-planar drawing respecting a given planar embedding~\cite{bertolazzi_flow,borradaile_planarFlow}. However, it is \NP-complete to decide whether a DAG admits an upward-planar drawing when no fixed embedding is given~\cite{GargTamassia01}. For special cases, upward-planarity testing in the variable embedding setting can be performed in polynomial time: e.g. if the DAG has only one source \cite{bertolazzi_singleSource,up-trees,hutton/lubiw:96},  or if the underlying undirected graph is series-parallel~\cite{Didimo_SPupward}. Furthermore, parameterized algorithms for upward-planarity testing exist with respect to the number of sources or the treewidth of the input DAG~\cite{DBLP:conf/compgeom/ChaplickGFGRS22}.

\begin{figure}[t]
	\centering
	\subcaptionbox{DAG $G$\label{SUBFIG:singleSourceExample}}{
		\begin{tikzpicture}[shorten <= -1pt, shorten >= -1pt,node distance=0.9cm]
			\node  (S) {$S$};
			\node (br) [above right of=S] {$v$}
			edge [<-] (S);
			\node (bl) [above left of=S] {$x$}
			edge [<-] (S);
			\node (N) [above left of=br] {$N$}
			edge [<-] (S)
			edge [<-] (br)
			edge [<-] (bl);
			\node (ur) [above right of=N] {$w$}
			edge [<-] (N)
			edge [<-] (br);
			\node (ul) [above left of=N] {$y$}
			edge [<-] (N)
			edge [<-] (bl);
			\node (t) [above of=ur] {$t$}
			edge [<-] (N)
			edge [<-,out=-20,in=-10] (S);
	\end{tikzpicture}}
	\hfil
	\subcaptionbox{Planar L-drawing\label{SUBFIG:LDrawing}}{
		\begin{tikzpicture}[scale=0.5]		
			\node (S) at (3,1) {$S$};
			\node (x) at (1,2) {$x$};
			\node (y) at (2,4) {$y$};
			\node (v) at (4,-1) {$v$};
			\node (w) at (5,0) {$w$};
			\node (N) at (6,3) {$N$};
			\node (t) at (7,-2) {$t$};
			\node (left) at (0,1) {};
			\node (right) at (8,7) {};
			\draw[rounded corners] (S) |- (x);
			\draw[rounded corners] (S) |- (N);
			\draw[rounded corners] (S) |- (v);
			\draw[rounded corners] (S) |- (x);
			\draw[rounded corners] (N) |- (y);
			\draw[rounded corners] (N) |- (t);
			\draw[rounded corners] (N) |- (w);
			\draw[rounded corners] (x) |- (y);
			\draw[rounded corners] (v) |- (w);
			\draw[rounded corners] (v) |- (N);
			\draw[rounded corners] (x) |- (N);
			\draw[rounded corners] (S) |- (t);
	\end{tikzpicture}}
	\hfil
	\subcaptionbox{Upward-planar L-drawing of \centerline{$G-\{S,t\}$}\label{SUBFIG:ULDrawing}}{
		\begin{tikzpicture}[scale=0.5]
			\node (S) at (3,1) {$S$};
			\node (x) at (1,2) {$x$};
			\node (y) at (2,5) {$y$};
			\node (v) at (6,3) {$v$};
			\node (w) at (5,6) {$w$};
			\node (N) at (4,4) {$N$};
			\node (t) at (7,7) {$t$};
			\node (left) at (0,1) {};
			\node (right) at (8,7) {};
			\draw[rounded corners] (S) |- (x);
			\draw[rounded corners] (S) |- (N);
			\draw[rounded corners] (S) |- (v);
			\draw[rounded corners] (S) |- (x);
			\draw[rounded corners] (N) |- (y);
			\draw[rounded corners] (N) |- (t);
			\draw[rounded corners] (N) |- (w);
			\draw[rounded corners] (x) |- (y);
			\draw[rounded corners] (v) |- (w);
			\draw[rounded corners] (v) |- (N);
			\draw[rounded corners] (x) |- (N);
	\end{tikzpicture}}
	\caption{(a) A single-source series-parallel DAG $G$. (b) A planar L-drawing of $G$. (c) An upward-planar L-drawing of $G$ without the edge $\{S,t\}$.}
\end{figure}
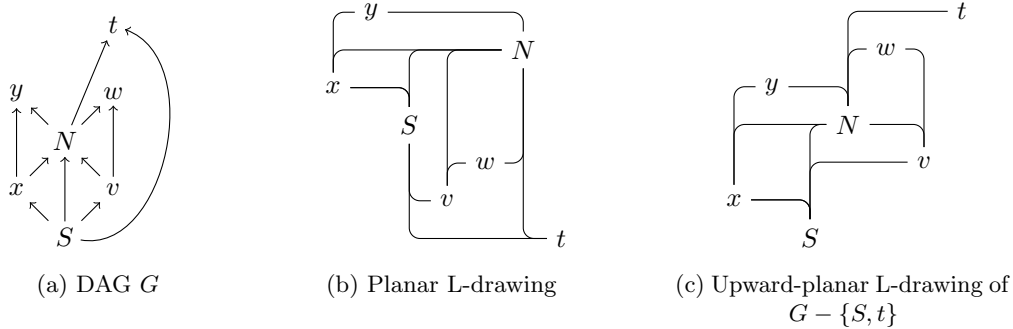

Every upward-planar DAG admits a straight-line upward-planar drawing
~\cite{DT-aprad-88}, however such a drawing may require exponential area~\cite{DBLP:journals/dcg/BattistaTT92}. Gronemann introduced bitonic \st-orderings for DAGs~\cite{gronemann:gd16}.
A plane \st-graph that admits a bitonic \st-ordering also admits an upward-planar drawing in quadratic area. It can be tested in linear time whether a plane \st-graph admits a bitonic \st-ordering~\cite{gronemann:gd16}, and whether a planar \st-graph admits a planar embedding that allows for a bitonic \st-ordering~\cite{angelini_etal:wg20,chaplick_etal:gd17}. By subdividing some transitive edges once, every plane \st-graph can be extended such that it admits a bitonic \st-ordering. 
Moreover, the minimum number of edges that have to be subdivided can be determined in linear time, both, 

In an \emph{L-drawing} of a directed graph~\cite{angeliniLdrawings} each~edge~$e$ is represented as a polyline composed of a vertical segment incident to the tail of $e$ and a horizontal segment incident to the head of $e$. In a \emph{planar L-drawing}, distinct edges may overlap, but not cross. 
See \cref{SUBFIG:LDrawing} for an example. 
The problem of testing for the existence of a planar L-drawing of a directed graph is \NP-complete~\cite{chaplick_etal:gd17}.
On the other hand, every upward-planar DAG admits a planar L-drawing \cite{angelini_etal:gd20}. A planar L-drawing is \emph{upward} if the lowest end vertex of the vertical segment of an edge $e$ is  the tail of $e$. See \cref{SUBFIG:ULDrawing} for an example. A planar \st-graph admits an upward-planar L-drawing if and only it admits a bitonic \st-ordering~\cite{chaplick_etal:gd17}.

\subparagraph{Our Contribution.} 
%
We characterize in \cref{THEO:augmentation} the plane DAGs admitting an upward-planar L-drawing as the subgraphs of plane \st-graphs admitting a bitonic \st-ordering. We first apply this characterization to prove that there are  trees with a fixed bimodal embedding that do not admit an upward-planar L-drawing (\cref{THEO:trees}). Moreover, the characterization allows to test in linear time whether any DAG with a single source or a single sink admits an upward-planar L-drawing preserving a given embedding (\cref{THEO:source}). 

We further show that every single-source acyclic cactus admits an upward-planar L-drawing by directly computing the x- and y-coordinates as post- and pre-order numbers, respectively, in a DFS-traversal (\cref{THEO:cacti}).  The respective result holds for single-sink acyclic cacti. Finally,
we use a dynamic-programming approach combined with a matching algorithm for regular expressions to decide in linear time whether a biconnected series-parallel DAG with a single source or a single sink has an embedding admitting an upward-planar L-drawing (\cref{THEO:variable}). Observe that a plane \st-graph does not necessarily admit an upward-planar L-drawing if the respective graph with reversed edges does. This justifies studying single-source and -sink graphs independently. 
Full details for proofs of statements marked with ($\star$) can be found in the Appendix. 

\section{Preliminaries}\label{SEC:preliminaries}

For standard graph theoretic notations and definitions we refer the reader to \cite{handbookGD}.

\subparagraph{Digraphs.} A \emph{directed graph (digraph)} $G=(V,E)$ is a pair consisting of a finite set $V$ of \emph{vertices} and a set $E$ of edges containing ordered pairs of distinct vertices. 
A vertex of a digraph is a \emph{source} if it is only incident to outgoing edges and a \emph{sink} if it is only incident to incoming edges. 
A \emph{walk} is a sequence of vertices  such that any two consecutive vertices in the sequence are adjacent.  A \emph{path} is a walk with distinct vertices. In this work we assume that all graphs are \emph{connected}, i.e., that there is always a path between any two vertices. A \emph{cycle} is a walk with distinct vertices  except for the first and the last vertex which must be equal. A \emph{directed path} (\emph{directed cycle}) is a path (cycle) where for any vertex $v$ and its successor $u$ in the path (cycle) there is an edge directed from $u$ to $v$. In the following, we only consider \emph{acyclic digraphs (DAGs)}, i.e., digraphs that do not contain directed cycles. 
A DAG is a \emph{tree} if it is connected and contains no cycles. It is a \emph{cactus} if it is connected and each edge is contained in at most one cycle. 

\subparagraph{Drawings.} In a \emph{drawing (node-link diagram)} of a digraph vertices are drawn as points in the plane and edges are drawn as simple curves between their end vertices. A drawing of a DAG is \emph{planar} if no two edges intersect except in common endpoints. A planar drawing splits the plane into connected regions~--~called \emph{faces}. A \emph{planar embedding} of a DAG is the counter-clockwise cyclic order of the edges around each vertex according to a planar drawing. A \emph{plane} DAG is a DAG with a fixed planar embedding and a fixed unbounded face. 

%

The \emph{rotation} of an orthogonal polygonal chain, possibly with overlapping edges, is defined as follows: We start with rotation zero. If the curve bends to the left, i.e., if there is a convex angle to the left of the curve, then we add $\pi/2$ to the rotation. If the curve bends to the right, i.e., if there is a concave angle to the left of the curve, then we subtract $\pi/2$ from the rotation.
Moreover, if the curve has a $2 \pi$ angle to the left, we handle this as two concave angles and if there is a 0 angle to the left, we handle this as two convex angles. 
The rotation of a simple polygon~--~with possible overlaps of consecutive edges~--~traversed in counterclockwise direction is $2 \pi$.

\subparagraph{Single-source series-parallel DAGs.}
\begin{figure}[t!]
	\centering
	\subcaptionbox{\label{SUBFIG:sssp-not-up}DAG not upward-planar}{
		\begin{tikzpicture}[shorten <= -2pt, shorten >= -2pt,node distance=0.8cm]
			\node (s) {$s$};
			\node (S) [above of=s] {$S$}
			edge [<-] (s);
			\node (u1)  [above left of=S] {$u_1$}
			edge [<-] (S);		
			\node (t1) [above of=u1] {$t_1$}
			edge [<-] (u1);
			\node (w1) [left of=t1] {$w_1$}
			edge [<-] (u1);
			\node (u2)  [above right of=S] {$u_2$}
			edge [<-] (S);		
			\node (t2) [above of=u2] {$t_2$}
			edge [<-] (u2);
			\node (w2) [right of=t2] {$w_2$}
			edge [<-] (u2);
			\node (N) [above right of=t1] {$N$}
			edge [<-] (w2)
			edge [->,shorten >=-4] (t2)
			edge [<-] (w1)
			edge [->,shorten >=-4] (t1);
			\node (dummy) [left of=w1] {};
			\node (v) [above of=N] {$v$}
			edge [->] (N);
			\draw [->]
			(s) .. controls (dummy) .. (v);
	\end{tikzpicture}}
	\hfil
	\subcaptionbox{\label{SUBFIG:decompositionTree}Decomposition-tree}{
		\begin{tikzpicture}[
			level 1/.style = {sibling distance = 2cm},
			level 3/.style = {sibling distance = 4cm},
			level 4/.style = {sibling distance = 2cm},
			level 6/.style = {sibling distance = 1cm},
			shorten <= -2pt, shorten >= -2pt,scale=0.4]
			\node {P}
			child {node {$Q$}} 
		child {node {S}
			child {node {$Q$}} 
		child {node {P}	
			child {node{S}
				child {node {$Q$}} 
			child {node {P}
				child {node {S}
					child {node {$Q$}} 
				child {node {$Q$}} 
		}
		child {node {S}
			child {node {$Q$}} 
		child {node {$Q$}} 
}
}
}
child {node{S}
child {node {$Q$}} 
child {node {P}
child {node {S}
child {node {$Q$}} 
child {node {$Q$}} 
}
child {node {S}
child {node {$Q$}} 
child {node {$Q$}} 
}
}
}
}
child {node {$Q$}} 
} ;
\end{tikzpicture}}
\hfil
\subcaptionbox{\label{SUBFIG:expanded-not-up}Embedding not upward-planar}{
\begin{tikzpicture}[shorten <= -2pt, shorten >= -2pt,node distance=0.8cm]
\node (s) {$s$};
\node (S) [above of=s] {$S$}
edge [<-] (s);
\node (t)  [above of=S] {$t$}
edge [<-] (S);		
\node (u) [left of=t] {$u$}
edge [<-] (S);
\node (w)  [right of=t] {$w$}
edge [<-] (S);
\node (N) [above of=t] {$N$}
edge [<-] (u)
edge [<-] (w)
edge [->] (t);
\node (v) [above of=N] {$v$}
edge [->] (N)
edge [<-, bend right=90] (s);
\end{tikzpicture}}
\caption{\label{FIG:sssp-not-up} a) A bimodal single-source series-parallel DAG that is not upward-planar b) with its decomposition-tree. c)  A single-source series-parallel DAG with an embedding that is not upward-planar. However, the DAG with a different embedding is upward-planar. }
\end{figure}
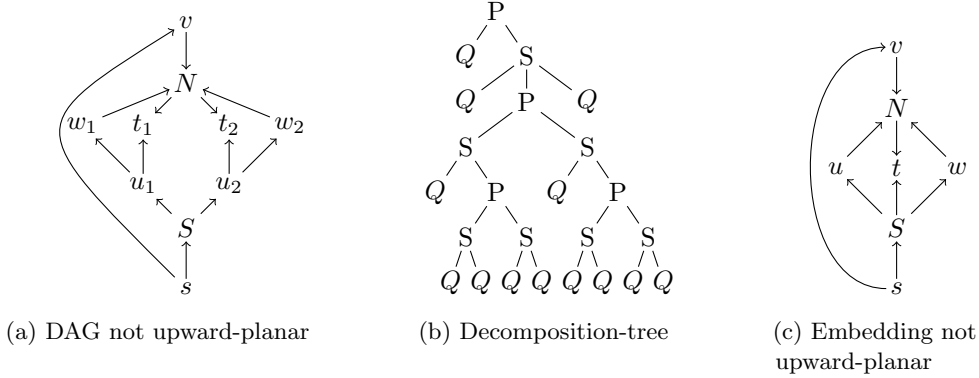
Series-parallel digraphs are digraphs with two distinguished vertices, called \emph{poles}, and can be defined recursively as follows: A single edge is a series-parallel digraph. Given $k$ series-parallel digraphs $G_1,\dots,G_k$ (\emph{components}), with poles $v_i,u_i$, $i=1,\dots,k$, a series-parallel digraph~$G$ with poles $v$ and $u$ can be obtained in two ways: by merging $v_1,\dots,v_k$ and $u_1,\dots,u_k$, respectively, into the new poles $v$ and $u$ (\emph{parallel composition}), or by merging the vertices $u_i$ and $v_{i+1}$, $i=1,\dots,k-1$, and setting $u=u_1$ and $v=v_{k}$ (\emph{series composition}).  The recursive construction of a series-parallel digraph is represented in a decomposition-tree $T$. We refer to the vertices of $T$ as \emph{nodes}. The leaves (vertices of degree one) of the decomposition-tree are labeled Q and represent the edges. The other nodes (\emph{inner nodes}) are labeled $P$ for parallel composition or $S$ for series composition. No two adjacent nodes of $T$ have the same label. \cref{SUBFIG:decompositionTree} shows the decomposition-tree of the graph in \cref{SUBFIG:sssp-not-up}. Let $\mu$ be a node of $T$. We denote by $T(\mu)$ the subtree rooted at $\mu$ and by $G(\mu)$ the subgraph of $G$ corresponding to $T(\mu)$, i.e., the subgraph of $G$ formed by the edges corresponding to the leaves of $T(\mu)$. The vertices of $G(\mu)$ that are different from its poles are called \emph{internal}. Given an arbitrary biconnected digraph $G$, it can be determined in linear time whether it is series-parallel, and a decomposition-tree of $G$ can be computed also in linear time \cite{gutwenger/mutzel:gd00}. 
Moreover, rooting a decomposition-tree of a biconnected series-parallel digraph $G$ at an arbitrary inner node yields again a decomposition-tree of $G$. 

In the following, we assume that $G$ is a series-parallel DAG with a single source (sink) $s$. If $G$ has more than one edge, we root the decomposition-tree $T$ at the inner node incident to the Q-node corresponding to an edge incident to $s$. This implies that for any node $\mu$ of $T$ no internal vertex of $G(\mu)$ can be a source (sink) of $G(\mu)$ and at least one of the poles of $G(\mu)$ is a source (sink) of $G(\mu)$. 

It follows from \cite{bertolazzi_singleSource} that every single-source series-parallel DAG is upward-planar if each vertex is incident to at most one incoming or at most one outgoing edge.  However, even in that case, not every bimodal embedding is already upward-planar, see \cref{SUBFIG:expanded-not-up}. 
Moreover, not every
single-source series-parallel DAG is upward-planar, even if it admits a bimodal embedding, see \cref{SUBFIG:sssp-not-up}. The reason for that is a P-node $\mu$ with two children $\mu_1$ and $\mu_2$ such that a pole $N$ of $G(\mu)$  is incident to an incoming edge in $G-G(\mu)$, and to both incoming and outgoing edges in both, $G(\mu_1)$ and $G(\mu_2)$. Bimodal single-source series-parallel DAGs without this property are always upward-planar~\cite{DBLP:journals/corr/abs-2205-05627}.

Given an upward-planar drawing of $G$ with distinct y-coordinates for the vertices, we call the pole of $G(\mu)$ with lower y-coordinate the \emph{South pole} of $G(\mu)$ and the other pole the \emph{North pole} of $G(\mu)$. Observe that the South (North) pole of $G$ is the unique source (sink)~$s$.  If $\mu$ is a $P$-node with children  $\mu_1,\dots,\mu_\ell$, then the South pole of $G(\mu_i)$, $i=1,\dots,\ell$ is the South pole of $G(\mu)$. Finally, if $\mu$ is an $S$-node with children  $\mu_1,\dots,\mu_\ell$, then observe that at most one among the components $G(\mu_i)$, $i=1,\dots,\ell$ can have more than one source (sink)~--~otherwise~$G$ would have more than one source (sink).  The South (North) pole of all other components is their unique source (sink). 

\subparagraph{Bitonic \st-ordering.}
A \emph{planar \st-graph} is a planar DAG with a single
source $s$, a single sink $t$, and an edge $(s,t)$. An \emph{\st-ordering} of a planar \st-graph is an enumeration~$\pi$
of the vertices with distinct integers, such that $\pi(u) < \pi(v)$
for every edge $(u,v)$. A \emph{plane \st-graph} is a planar \st-graph with a planar embedding in which the edge $(s,t)$ is incident to the outer face. Every plane \st-graph admits an upward-planar drawing~\cite{DT-aprad-88}. 

For each vertex $v$ of a plane \st-graph, we consider the ordered list $S(v)=\left< v_1,v_2,\dots,v_k \right>$ of
the successors of $v$ as they appear from left to right in an upward-planar drawing. 
An \st-ordering of a plane \st-graph is \emph{bitonic},  if there is a vertex $v_h$ in $S(v)=\left< v_1,v_2,\dots,v_k \right>$ such that
$\pi(v_i)<\pi(v_{i+1})$, $i=1,\dots,h-1$, and $\pi(v_i)>\pi(v_{i+1})$,
$i=h,\dots,k-1$. 
%
We say that the successor list $S(v)=\left< v_1,v_2,\dots,v_k \right>$ of a vertex $v$ contains a \emph{valley} if there are $1 < i \leq j < k$ such that there are both, a directed $v_i$-$v_{i-1}$-path and a directed $v_j$-$v_{j+1}$-path in $G$.
See \cref{FIG:char}. Gronemann~\cite{gronemann:gd16} characterized the plane \st-graphs that admit a bitonic \st-ordering as follows.

\begin{theorem}[\cite{gronemann:gd16}]\label{THEO:char}
A plane \st-graph admits a bitonic \st-ordering if and only if
the successor list of no vertex contains a valley. 
\end{theorem}

\begin{figure}
\centering
\subcaptionbox{Valley\label{SUBFIG:valley}}[0.45\textwidth]{
\begin{tikzpicture}
\node  (1) {$v_{i-1}$};
\node (a) [below right of=1] {};
\node (2) [right of=a] {$v_i$};
\draw [->,decorate,decoration={snake}] (2) -- (1);
\node (dots) [right of=2] {$\dots$};
\node (3) [right of=dots] {$v_j$};
\node (b) [above right of=3] {};
\node (4) [right of=b] {$v_{j+1}$};
\draw [->,decorate,decoration={snake}] (3) -- (4);
\node (x) [below of=dots] {};
\node (u) [below of=x] {$v$}
edge [->] (1)
edge [->] (2)
edge [->] (3)
edge [->] (4);
\end{tikzpicture}}%
\hfil
\subcaptionbox{Fixed\label{SUBFIG:fixed}}{
\begin{tikzpicture}
\node  (1) {$\bullet$};
\node (2) [above right of=1] {};
\node (3) [above left of=1] {};
\node (4) [above left of=2] {$\bullet$}
edge [<-] (1);
\node (5) [above right of=4] {$\circ$}
edge [<-] (4)
edge [<-] (1);
\node (6) [above left of=4] {$\circ$}
edge [<-] (4)
edge [<-] (1);
\end{tikzpicture}}%
\hfil
\subcaptionbox{Variable\label{SUBFIG:variable}}{
\begin{tikzpicture}[shorten <= -3pt, shorten >= -3pt,scale=0.4]
\node  (S) at (0,0) {$\circ$};
\node (br) at (-1,2) {$\circ$}
edge [<-] (S);
\node (bl) at (-2,1) {$\circ$}
edge [<-] (S);
\node (N) at (-3,3) {$\circ$}
edge [<-] (S)
edge [<-] (br)
edge [<-] (bl);
\node (ur) at (-4,5) {$\circ$}
edge [<-] (N)
edge [<-] (br);
\node (ul) at (-5,4) {$\circ$}
edge [<-] (N)
edge [<-] (bl);
\node (t) at (0,6) {$\circ$};
\draw [<-,shorten <= 3 pt, shorten >= 3pt] (0,6) .. controls (-1,6) and (-7,6) .. (-3,3); 
\node (br2) at (1,2) {$\circ$}
edge [<-] (S);
\node (bl2) at (2,1) {$\circ$}
edge [<-] (S);
\node (N2) at (3,3) {$\circ$}
edge [<-] (S)
edge [<-] (br2)
edge [<-] (bl2);
\node (ur2) at (4,5) {$\circ$}
edge [<-] (N2)
edge [<-] (br2);
\node (ul2) at (5,4) {$\circ$}
edge [<-] (N2)
edge [<-] (bl2);
\draw [<-,shorten <= 3pt, shorten >= 3pt] (0,6) .. controls (1,6) and (7,6) .. (3,3); 
\end{tikzpicture}}

\caption{\label{FIG:char}(a) Forbidden configuration for bitonic $st$-orderings. (b+c) Single-source series-parallel plane DAG that does not admit an upward-planar L-drawing since it contains a valley, in  case (c) in any upward-planar embedding. }
\end{figure}
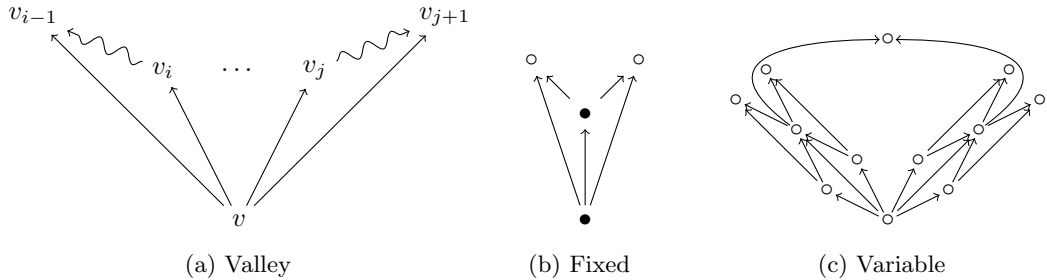

\section{Upward-Planar L-Drawings~--~A Characterization}\label{SEC:L-drawings}

A plane \st-graph admits an upward-planar L-drawing if and only it admits a bitonic
\st-ordering~\cite{chaplick_etal:gd17}. We extend this result to general plane DAGs and discuss some consequences. 

\begin{theorem}\label{THEO:augmentation}
A plane
DAG admits an upward-planar L-drawing if and only if it can be
augmented to a plane 
\st-graph that admits an
upward-planar L-drawing, i.e., a plane \st-graph that admits a bitonic
\st-ordering.
\end{theorem}
\begin{proof}
Let $G$ be a plane DAG. Clearly, if an augmentation of $G$ admits an upward-planar L-drawing, then so does $G$.
Let now an upward-planar drawing of $G$ be given. Add a
directed triangle with a new source $s$, a new sink $t$, and a new vertex $x$ enclosing 
the drawing of $G$. 
As long as there is a vertex $v$ of $G$ that is not incident to an incoming or outgoing edge,
shoot a ray from 
$v$ to the top or the right, respectively,
until it hits another edge and follow the segment to the incident
vertex~--~recall that one end of any segment is a vertex and one end is a
bend. The orientation of the added edge is implied by the L-drawing.  
The result is an upward-planar L-drawing of a digraph with the single source $s$ and the single sink~$t$.
\end{proof}

Observe that every series-parallel \st-graph admits a bitonic \st-ordering \cite{angelini_etal:wg20,chaplick_etal:gd17} and, thus, an upward-planar L-drawing. This is no longer true for upward-planar series-parallel DAGs with several sources or several sinks. \cref{SUBFIG:fixed,SUBFIG:variable} show examples of two single-source upward-planar series-parallel DAGs that contain a valley. 
There are even upward-planar series-parallel DAGs with a single source or a single sink that do not admit an upward-planar L-drawing, even though the successor list of no vertex contains a valley.

\begin{figure}[t!]
\centering
\subcaptionbox{\label{SUBFIG:singleSink}Single sink}{
\begin{tikzpicture}
\node  (S) {$\circ$};
\node (w) [below right of=S] {$w$}
edge [->] (S);
\node (N) [above right of=w] {$\circ$}
edge [<-] (w);
\node  (v) [below of=w] {$v$}
edge [->,dashed] (w)
edge [->] (S)
edge [->] (N);
\node (t) [above right of=S] {$t$}
edge [<-] (S)
edge [<-] (N);
\node (left) [left of=S] {};
\node (right) [right of=N] {};
\end{tikzpicture}}
\hfil
\subcaptionbox{\label{SUBFIG:singleSource}Single source}{
\begin{tikzpicture}
\node  (S) {$S$};
\node (br) [above right of=S] {$v$}
edge [<-] (S);
\node (bl) [above left of=S] {$\circ$}
edge [<-] (S);
\node (N) [above left of=br] {$N$}
edge [<-] (S)
edge [<-] (br)
edge [<-] (bl);
\node (ur) [above right of=N] {$w$}
edge [<-] (N)
edge [<-] (br);
\node (ul) [above left of=N] {$\circ$}
edge [<-] (N)
edge [<-] (bl);
\node (t) [above of=ur] {$t$}
edge [<-, dashed] (ur)
edge [<-] (N)
edge [<-,out=-20,in=-10] (S);
\end{tikzpicture}} \hfil
\subcaptionbox{\label{SUBFIG:tree}Family of trees}{
\begin{tikzpicture}[shorten <= -1.6pt, shorten >= -2pt, node distance=0.4cm]
\node (s7) {$\circ$};
\node (7) [above right of = s7] {$\circ$}
edge [->] (s7);
\node (s5) [above right of = 7] {$\circ$};
\node (5) [above right of = s5] {$\circ$}
edge [->] (s5);
\node (s3) [above right of = 5] {$\circ$};
\node (3) [above right of = s3] {$\circ$}
edge [->] (s3);
\node (s1) [above right of = 3] {$\circ$};
\node (1) [above right of = s1] {$\circ$}
edge [->] (s1);
\node (2) [above right of = 1] {$\circ$}
edge [<-, bend right] (3)
edge [->] (1);
\node (t2) [above right of = 2] {$\circ$}
edge [->] (2);
\node (4) [above right of = t2] {$\circ$}
edge [<-, bend right] (5)
edge [<-, bend left] (3);
\node (t4) [above right of = 4] {$\circ$}
edge [->] (4);
\node (6) [above right of = t4] {$\circ$}
edge [<-, bend right] (7)
edge [<-, bend left] (5);
\node (t6) [above right of = 6] {$\circ$}
edge [->] (6);
\node (dots) [right of = 7] {$\dots$};
\end{tikzpicture}}
\caption{DAGs that do not admit an upward-planar L-drawing even though they do not contain a valley. Dashed edges indicate augmentations and are not part of the DAG.}
\end{figure}
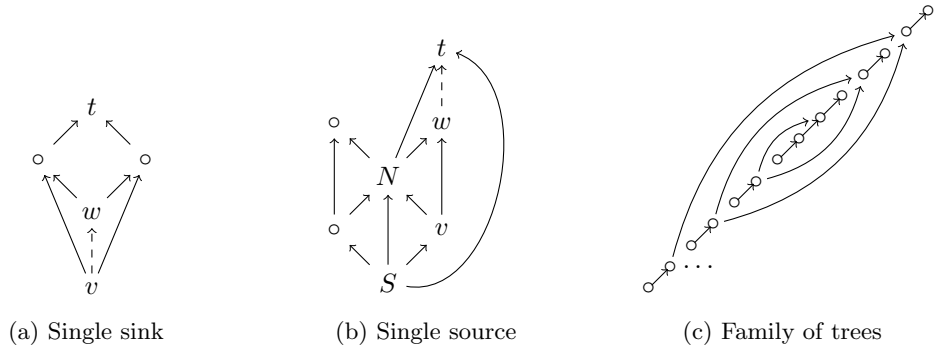

Consider the DAG $G$ in \cref{SUBFIG:singleSink} (without the dashed edge). $G$ has a unique upward-planar embedding. Since no vertex has more than two successors there cannot be a valley. Assume $G$ admits a planar L-drawing. By \cref{THEO:augmentation} there should be an extension of $G$ to a plane \st-graph $G'$ that admits a bitonic \st-ordering. But the internal source $w$ can only be eliminated by adding the edge $(v,w)$. Thus $w$ is a successor of $v$ in $G'$. Hence, the successor list of $v$ in $G'$ contains a valley. By \cref{THEO:char}, $G'$ is not bitonic, a contradiction.

Now consider the DAG $G$ in \cref{SUBFIG:singleSource} (without the dashed edge). $G$ has two symmetric upward-planar embeddings: with the curved edge to the right or the left of the remainder of the DAG. We may assume that the curved edge is to the right. But then an augmentation to a plane \st-graph $G'$ must contain the dashed edge, which completes a valley at the single source and its three rightmost outgoing edges. By \cref{THEO:char}, $G'$ is not bitonic, a contradiction.

%

A planar L-drawing is \emph{upward-leftward} \cite{chaplick_etal:gd17} if all edges are upward and point to the left. 

\begin{theorem}[Trees]\label{THEO:trees} Every directed tree admits an upward-leftward planar
L-drawing, but not every tree with a fixed bimodal embedding admits an upward-planar L-drawing.
\end{theorem}
\begin{proof}
If the embedding is not fixed, we can construct an upward-planar
L-drawing of the input tree by removing one leaf $v$ and its incident edge $e$, drawing the smaller
directed tree inductively, and inserting the removed leaf into this upward-leftward planar L-drawing. To this end let $u$ be the unique neighbor of $v$. We embed $e$ as the first incoming or outgoing edge of $u$, respectively, in counterclockwise direction, and draw $v$ slightly to the right and below $u$, if $e$ is an incoming edge of $u$, or slightly to the left and above $u$, if $e$ is an outgoing edge of $v$. 
This guarantees that the resulting L-drawing is upward-leftward and planar.

When the embedding is fixed, we consider a family of plane trees $T_k$, $k \geq 1$, proposed by Frati \cite[Fig.~4a]{fratiJCGA08}, that have $2k$ vertices and require an exponential area $\Omega(2^{k/2})$ in any embedding-preserving straight-line upward-planar drawing; see \cref{SUBFIG:tree}. We claim that, for sufficiently large $k$, the tree $T_k$ does not admit an upward-planar L-drawing. Suppose, for a contradiction, that it admits one. By \cref{THEO:augmentation}, we can augment this drawing to an upward-planar L-drawing of a plane \st-graph $G$ with $n=2k+3$ vertices. This implies that $G$ admits a bitonic \st-ordering \cite{chaplick_etal:gd17}. Hence, $G$ (and thus $T_k$) admits a straight-line upward-planar drawing in quadratic area  $(2n - 2) \times (n - 1)$ \cite{gronemann:gd16}, a contradiction. 
\end{proof}


\section{Single-Source or -Sink DAGs with Fixed Embedding}\label{SEC:fixed} 

In the fixed embedding scenario,
we first prove that every single-source or -sink acyclic cactus with no transitive edge admits an upward-planar L-drawing and then give a linear-time algorithm to test whether a single-source or -sink DAG admits an upward-planar L-drawing. 


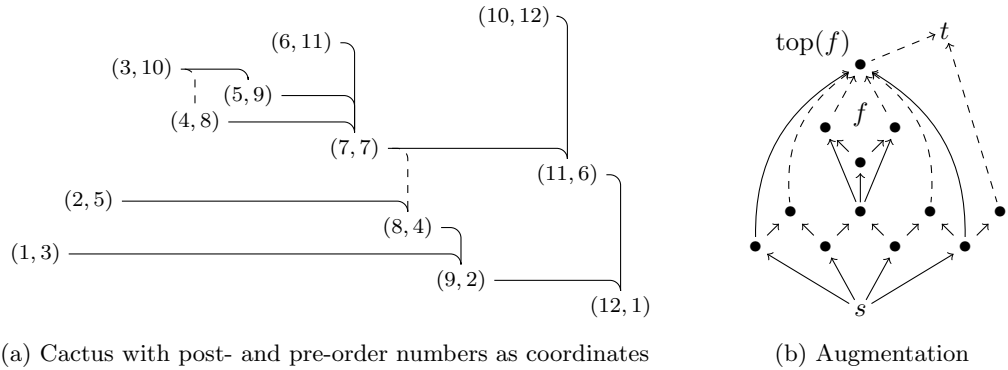
\begin{figure}[t]
\centering
\subcaptionbox{\label{FIG:cactus}Cactus with post- and pre-order numbers as coordinates}{
\footnotesize
\begin{tikzpicture}[scale=0.7]
\node (12-1) at (12,0.5) {$(12,1)$};
\node (9-2) at (9,1) {$(9,2)$};
\node (8-4) at (8,2) {$(8,4)$};
\node (1-3) at (1,1.5) {$(1,3)$};
\node (2-5) at (2,2.5) {$(2,5)$};
\node (7-7) at (7,3.5) {$(7,7)$};
\node (11-6) at (11,3) {$(11,6)$};
\node (10-12) at (10,6) {$(10,12)$}; 
\node (4-8) at (4,4) {$(4,8)$};
\node (3-10) at (3,5) {$(3,10)$};
\node (5-9) at (5,4.5) {$(5,9)$};
\node (6-11) at (6,5.5) {$(6,11)$};
\draw[rounded corners] (12-1) |- (9-2);			
\draw[rounded corners] (9-2) |- (1-3);			
\draw[rounded corners] (9-2) |- (8-4);			
\draw[rounded corners] (8-4) |- (2-5);			
\draw[rounded corners,dashed] (8-4) |- (7-7);			
\draw[rounded corners] (12-1) |- (11-6);			
\draw[rounded corners] (11-6) |- (10-12);			
\draw[rounded corners] (11-6) |- (7-7);			
\draw[rounded corners] (7-7) |- (4-8);			
\draw[rounded corners,dashed] (4-8) |- (3-10);			
\draw[rounded corners] (7-7) |- (5-9);			
\draw[rounded corners] (5-9) |- (3-10);			
\draw[rounded corners] (7-7) |- (6-11);				
\end{tikzpicture}
} \hfil
\subcaptionbox{\label{FIG:augmentation}Augmentation}{
\begin{tikzpicture}
[shorten <= -2pt, shorten >= -2pt, node distance=0.65cm]
\node (1) {$\bullet$};
\node (2) [above right of = 1] {$\bullet$}
edge [<-] (1);
\node (3) [below right of = 2] {$\bullet$}
edge [->] (2);
\node (4) [above right of = 3] {$\bullet$}
edge [<-] (3);
\node (5) [below right of = 4] {$\bullet$}
edge [->] (4);
\node (6) [above right of = 5] {$\bullet$}
edge [<-] (5);
\node (7) [below right of = 6] {$\bullet$}
edge [->] (6);
\node (b) [above of = 4] {$\bullet$}
edge [<-] (4);
\node (a) [above left of = b] {$\bullet$}
edge [<-] (4)
edge [<-] (b);
\node (c) [above right of = b] {$\bullet$}
edge [<-] (4)
edge [<-] (b);
\node (f) [above of = b] {$f$};
\node (top) [above of = f] {$\bullet$}
edge [<-,dashed] (a)
edge [<-,dashed] (c)
edge  [<-, bend right, dashed] (2)
edge  [<-, bend left, dashed] (6)
edge  [<-, bend right] (1)
edge  [<-, bend left] (7);
\node (label) [above left = -3mm of top] {top$(f)$};
\node (dummy) [below of = 4] {};
\node (s) [below of = dummy] {$s$}
edge [->] (1)
edge [->] (3)
edge [->] (5)
edge [->] (7);
\node (x) [above right of = 7] {$\bullet$}
edge [<-] (7);
\node (dummy2) [right of = top] {};
\node (t) [above right of = dummy2] {$t$}
edge [<-,dashed] (top)
edge [<-,dashed] (x);
\end{tikzpicture}
}

\caption{\label{FIG:dashed} (a) Single-source acyclic cactus. The dashed edges are the last edges on a left path cycles. (b) A new sink $t$ and the dashed edges augment a plane single-source DAG~to~a~plane~\st-graph. }
\end{figure}
\begin{theorem}[Plane Single-Source or Single-Sink Cacti]\label{THEO:cacti} Every acyclic cactus $G$ with a single source or single sink admits an upward-leftward outerplanar L-drawing. Moreover, if there are no transitive edges (e.g., if $G$ is a tree) then such a drawing can be constructed so to maintain a given outerplanar embedding.
\end{theorem}
\begin{proof}
We first consider the case that $G$ has a single source $s$. Observe that then each biconnected component $C$ of $G$ (which is either an edge or a cycle) has a single source, namely the cut-vertex of $G$ that separates it from the part of the DAG containing $s$. This implies that $C$ also has a single sink (although $G$ may have multiple sinks, belonging to different biconnected components). In particular, if $C$ is a cycle, it consists of a \emph{left} path $P_\ell$ and a \emph{right} path $P_r$ between its single source and single sink. By flipping the cycle $C$~--~maintaining outerplanarity~--~we can ensure that $P_\ell$ contains more than one edge. Note that this flipping is only performed if there are transitive edges.
Consider the tree $T$ that results from $G$ by removing the last edge of every left path. 

We perform a depth-first search on $T$ starting from $s$ where the edges around a vertex are traversed in clockwise order. 
We enumerate each vertex twice, once when we first meet it (DFS-number or \emph{preorder} number) and once when we finally backtrack from it (\emph{postorder} number). To also obtain that each edge points to the left, backtracking has to be altered from the usual DFS: Before backtracking on a left path $P_\ell$ of a cycle $C$, we directly jump to the single source $s_C$ of $C$ and continue the DFS from there, following the right path $P_r$~of~$C$. Only once we have backtracked from the single sink $t_C$ of $C$, we give each vertex on $P_\ell$, excluding $s_C$, a postorder number and then we continue backtracking on $P_r$. See \cref{FIG:cactus}.

Let the y-coordinate of a vertex be its preorder number and let the x-coordinate be its thus constructed postorder number. Since each vertex has a larger
preorder- and a lower postorder-number than its parent,
the drawing is upward-leftward.  In~\cite{DBLP:journals/corr/abs-2205-05627} we prove that it is also planar and preserves the embedding, which was updated only in the presence of transitive edges.

Now consider the case that $G$ has a single sink. Flip the embedding, i.e., reverse the linear order of the incoming (outgoing) edges around each vertex.
Reverse the orientation of the edges, construct the drawing of the resulting single-source DAG, rotate it by 90 degrees in counter-clockwise direction, and mirror it horizontally.
This yields the desired drawing.
\end{proof}
\subparagraph{General DAGs.}
Two consecutive incident edges of a vertex form an \emph{angle}.  A \emph{large angle} in an upward-planar straight-line drawing is an angle greater than $\pi$ between two consecutive edges incident to a source or a sink, respectively. An \emph{upward-planar embedding} of an upward-planar DAG is a planar embedding with the assignment of large angles according to a straight-line upward-planar drawing. For single-source or single-sink DAGs, respectively, a planar embedding and a fixed outer face already determine an upward-planar embedding~\cite{bertolazzi_singleSource}.

An angle is a  \emph{source-switch} or a \emph{sink-switch}, respectively, if the two edges are both outgoing or both incoming edges of the common end vertex. Observe that the number~$A(f)$ of source-switches in a face~$f$ equals the number of sink-switches in~$f$. Bertolazzi et al.~\cite{bertolazzi_flow} proved that in biconnected upward-planar DAGs, the number $L(f)$ of large angles in a face~$f$ is $A(f)-1$, if $f$ is an inner face, and $A(f)+1$, otherwise, and mentioned in the conclusion that this result could be extended to simply connected graphs. An explicit proof for single-source or -sink DAGs can be found in~\cite{DBLP:journals/corr/abs-2205-05627}.

\begin{theorem}\label{THEO:source}
Given an upward-plane  single-source or single-sink DAG, it can be tested in linear time whether it admits an upward-planar L-drawing.
\end{theorem}

In the following, we prove the theorem for a DAG $G$ with a single source $s$;
the single-sink case is discussed in~\cite{DBLP:journals/corr/abs-2205-05627}.
%
%
In an upward-planar straight-line drawing of $G$, the only large angle at a source-switch is the angle at $s$ in the outer face. Thus, in the outer face all angles at sink-switches are large and in an inner face $f$ all but one angle at sink-switches are large. For an inner face $f$, let top$(f)$ be the sink-switch of $f$ without large angle. See \cref{FIG:augmentation}.

\begin{lemma}\label{LEMMA:path-to-top}
Let $G$ be a single source upward-planar DAG with a fixed upward-planar embedding, let $f$ be an inner face, and let $v$ be a sink with a large angle in $f$. Every plane \st-graph extending $G$ contains a directed $v$-top($f$)-path.
\end{lemma}

\begin{proof}
Consider a planar $st$-graph extending $G$. In this graph there must be an outgoing edge $e$ of $v$ towards the interior of $f$. 
Let $w$ be the head of~$e$. Follow a path from $w$ on the boundary of $f$ upward until a sink-switch $v'$ is met. 
If this sink-switch is top$(f)$, we are done. 
Otherwise continue recursively by considering an outgoing edge $e'$ of $v'$ toward the interior~of~$f$.  Eventually this process terminates when top$(f)$ is reached.
\end{proof}

\begin{proof}[Proof of \cref{THEO:source}, single-source case]
Let $G$ be an upward-planar single-source DAG with a fixed upward-planar embedding. Let $G'$ be the DAG that results from $G$ by adding  in each inner face $f$ edges from all sinks with a large angle in $f$ to top$(f)$ and by adding a new sink $t$ together with edges from all sink-switches on the outer face. We will show that $G$ admits an upward-planar L-drawing if and only if $G'$ does.
This implies the statement, since testing whether $G'$ admits a bitonic \st-ordering can be performed in linear time~\cite{gronemann:gd16}.

Clearly, if $G' \supseteq G$ admits an upward-planar L-drawing, then so does $G$. 
In order to prove the other direction, suppose that $G$ has an upward-planar L-drawing. In order to prove that $G'$ admits an upward-planar L-drawing, we show that it is a planar $st$-graph that admits a bitonic $st$-ordering~\cite{chaplick_etal:gd17}. To show this, we argue that 
$G'$ is acyclic, has a single source and a single sink, and the successor list of no vertex contains a valley by \cref{THEO:char}.  

Namely, the edges to the new sink $t$ cannot be contained in any directed cycle.
Furthermore, by \cref{THEO:augmentation},
there is an augmentation $G''$ of $G$ such that (a) $G''$ is a planar \st-graph and such that (b) $G''$ admits an upward-planar L-drawing.
By \cref{LEMMA:path-to-top}, the edges added into inner faces of $G$ either belong to $G''$ or are transitive edges in $G''$. Thus, $G'$ is acyclic. 

Since $G'$ does not have more sources than $G$, there is only one source in $G'$. Each sink has a large angle in some face. Thus, in $G'$ each vertex other than $t$  has at least one outgoing~edge. Therefore, $G'$ is a planar $st$-graph.

Assume now that there is a face $f$ with a sink $w$ such that the edge $(w,$top$(f))$ would be part of a valley at a vertex $v$ in $G'$, i.e., assume there are successors $v_{i-1},v_{i},v_j,v_{j+1}$ of~$v$ from left to right (with possibly $v_i= v_j$) such that there is both, a directed $v_i$-$v_{i-1}$-path and a directed $v_j$-$v_{j+1}$-path. Since the out-degree of $w$ in $G'$ is one, it follows that $w \neq v$. Thus, $(w,$top$(f))$  could only be part of the $v_i$-$v_{i-1}$-path or the $v_j$-$v_{j+1}$-path. But then, by \cref{LEMMA:path-to-top}, there would be such a path in any augmentation of $G$ to a planar \st-graph. 
Finally, the edges incident to $t$ cannot be involved in any valley, since all the tails have out-degree~1. Thus, $G'$ contains no valleys.
\end{proof}

\section{Single-Source~or~-Sink~Series-Parallel~DAGs~with~Variable~Embedding}
\label{SEC:variable}

The goal of this section is to prove the following theorem.

\begin{theorem}\label{THEO:variable}
It can be tested in linear time whether a biconnected series-parallel DAG with a single source or a single sink admits an upward-planar L-drawing.
\end{theorem}



In the following we assume that $G$ is a biconnected series-parallel DAG. 

\subparagraph{Single Source.}
We follow a dynamic-programming approach inspired by Binucci et al.~\cite{binucci_etal:socg19} and Chaplick et al.~\cite{DBLP:conf/compgeom/ChaplickGFGRS22}.
We define feasible types that combinatorially describe the ``shapes'' attainable in an upward-planar L-drawing of each component. We show that these types are sufficient to determine the possible types of a graph obtained with a parallel or series composition, and show how to compute them efficiently. 
The types depend on the choice of the South pole as the bottommost pole (if it is not uniquely determined by the structure of the graph, e.g., if one of them is the unique source),  and on the type of the leftmost $S$-$N$-path $P_L$ and the rightmost $S$-$N$-path $P_R$ between the South-pole $S$ and the North-pole~$N$. Observe that $P_L$ and $P_R$ do not have to be directed paths.

\begin{figure}
\centering
\subcaptionbox{$L$\label{SUBFIG:southL}}{\includegraphics[page=1]{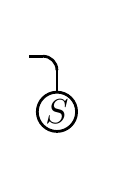}} \hfil
\subcaptionbox{$R$\label{SUBFIG:southR}}{\includegraphics[page=2]{Figures/poles}} \hfil
\subcaptionbox{$R^{cc}$\label{SUBFIG:northRcc}}{\includegraphics[page=3]{Figures/poles}} \hfil
\subcaptionbox{$L^{cc}$\label{SUBFIG:northLcc}}{\includegraphics[page=4]{Figures/poles}} \hfil
\subcaptionbox{$W$\label{SUBFIG:northW}}{\includegraphics[page=5]{Figures/poles}} \hfil
\subcaptionbox{$E$\label{SUBFIG:northE}}{\includegraphics[page=6]{Figures/poles}} \hfil
\subcaptionbox{$R^c$\label{SUBFIG:northRc}}{\includegraphics[page=7]{Figures/poles}} \hfil
\subcaptionbox{$L^c$\label{SUBFIG:northLc}}{\includegraphics[page=8]{Figures/poles}}
\caption{The different types of a path between the poles. (a,b) South types; (c-h) North types.
}
\label{fig:pole-types}
\end{figure}

More precisely, the type of an $S$-$N$-path $P$ is defined as follows: There are two \emph{South-types} depending on the edge incident to $S$: \emph{$L$} (outgoing edge bending to the \emph{left}; \cref{SUBFIG:southL}) and \emph{$R$} (outgoing edge bending to the \emph{right}; \cref{SUBFIG:southR}).
For the last edge incident to the North pole $N$ we have in addition the types for the incoming edges: \emph{$W$} (incoming edge entering from the left~--~\emph{West}; \cref{SUBFIG:northW}) and \emph{$E$} (incoming edge entering from the right~--~\emph{East}; \cref{SUBFIG:northE}).
For the types $R$ and $L$ we further distinguish whether $P$  passes to the left of $N$ (\emph{$R^{cc}$}/\emph{$L^{cc}$}; \cref{SUBFIG:northRcc,SUBFIG:northLcc}) or to the right of $N$ (\emph{$R^{c}$}/\emph{$L^{c}$};  \cref{SUBFIG:northRc,SUBFIG:northLc}): Let $h$ be the horizontal line through $N$.
We say that $P$ \emph{passes to the left (right)} of $N$ if the last edge of $P$ (from $S$ to $N$) that intersects $h$ does so to the left (right) of $N$.
Thus, 
there are six \emph{North-types} for a path between the poles: $R^{cc},L^{cc},W,E,R^c,L^c$.
%
The superscripts $c$~and~$cc$ stand for clockwise and counter-clockwise, respectively, to denote the rotation of a path that passes to the left (right) of $N$, when walking from $N$ to $S$. This is justified in the next lemma and depicted  e.g., in \cref{SUBFIG:ElLr}, where the right $S$-$N$-path has type $L^c$, since (walking from~$N$ to~$S$) it first bends to the left and then passes to the right of $N$ by rotating clockwise.


\begin{restatable}[$\star$]{lemma}{rotation}
Let $G$ be a series-parallel DAG with no internal sources. Let an upward-planar L-drawing of $G$ be given where the poles $S$ and $N$ are incident to the outer face and~$S$ is below $N$.
Let $P$ be a not necessarily directed $S$-$N$-path. 
%
Let $P'$ be the polygonal chain obtained from $P$ by adding a vertical segment pointing from $N$ downward. The rotation~of~$P'$~is 
\begin{itemize}
\item $\pi$ if the type of $P$ at $N$ is in $ \{E,L^c,R^c\}$.
\item $-\pi$ if the type of $P$ at $N$ is  in $ \{L^{cc},R^{cc},W\}$.
\end{itemize} 
\end{restatable}

\medskip

We say that the \emph{type of a path between the poles} is $(X,x)$, if $X$ is the North-type and $x$ is the South-type of the path, e.g., the type of a path that bends right at the South-pole, passes to the right of the North-pole and ends in an edge that leaves the North-pole bending to the left is $(L^c,R)$, see $P_R$ in \cref{SUBFIG:ElLr}.
For two North-types $X$ and $Y$, we say $X < Y$ if $X$ is before $Y$ in the ordering  $[R^{cc}L^{cc}WER^cL^c]$. The South-types are ordered $L < R$.
For two types $(X,x)$ and $(Y,y)$ we say that $(X,x) \leq (Y,y)$ if $X \leq Y$ and $x \leq y$, and $(X,x)  < (Y,y)$ if  $(X,x) \leq (Y,y)$ and $X < Y$ or $x < y$.

The \emph{type of a component} is determined by eight entries, whether the component is a single edge or not, the choice of the bottommost pole (South pole), the type of $P_L$, the type of $P_R$, and additionally four \emph{\textsc{free}-flags}: For each pole, two flags \emph{left-free} and \emph{right-free} indicating whether the bend on $P_L$ and $P_R$, respectively, on the edge incident to the pole  is \emph{free} on the left or the right, respectively: More precisely, let $P$ be an $S$-$N$-path and let $e$ be an edge of~$P$ incident to a pole~$X$. We say that $e$ is \emph{free} on the right (left) at $X$ if $e$ bends to the right (left)~--~walking from $S$ to $N$~--~or if the bend on $e$ is not contained in an edge not incident to~$X$. See \cref{SUBFIG:series-counterclockwise,SUBFIG:series-clockwise}.
We denote a type by $\left<(X,x),(Y,y)\right>$ where $(X,x)$ is the type of $P_L$ and $(Y,y)$ is the type of $P_R$ without explicitly mentioning the flags or the choice of the South pole.  Observe that $Y < L^c$ if $X = R^{cc}$ and $\left<(X,x),(Y,y)\right>$ is the type of a component.
\cref{FIG:parallel} illustrates how components of different types can be composed in parallel.

%

%
%

\begin{figure}
\centering
\subcaptionbox{$\left<(E,L),(L^c,R)\right>$\label{SUBFIG:ElLr}}[.2\textwidth]{\includegraphics[page=13]{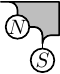}} \hfil
\subcaptionbox{$\left<(R^{cc},R),(R^{cc},R)\right>$\label{SUBFIG:series-counterclockwise}}{\includegraphics[page=3,scale=0.85]{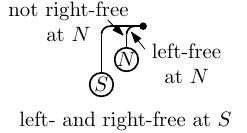}} \hfil
\subcaptionbox{$\left<(R^{c},R),(R^{c},R)\right>$\label{SUBFIG:series-clockwise}}{\includegraphics[page=2,scale=0.85]{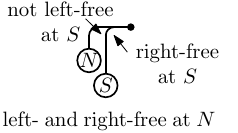}} \hfil
\subcaptionbox{Parallel composition\label{FIG:parallel}}{\includegraphics[page=3]{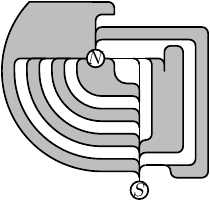}}

\caption{ (a--c) Illustrations for types of a component. (d) A parallel composition of eight components of the following types: 
$\left<(R^{cc},L),(W,L)\right>$,
$\left<(W,L),(W,L)\right>$,
$\left<(W,L),(W,L)\right>$,
$\left<(W,L),(E,L)\right>$,
$\left<(E,L),(E,L)\right>$ single edge,
$\left<(E,R),(E,R)\right>$ not left-free at $N$,
$\left<(R^c,R),(R^c,R)\right>$. The result is of type $\left<(R^{cc},L),(R^c,R)\right>$.
}
\end{figure}

\begin{restatable}[Parallel Composition ($\star$)]{lemma}{parallelComposition}\label{LEMMA:parallelComposition}

A component $C$ of type $\left<(X,x),(Y,y)\right>$ with the given four \textsc{free}-flags can be obtained as a parallel composition of components $C_1,\dots,C_\ell$ of type  \linebreak
$\left<(X_1,x_1),(Y_1,y_1)\right>,\dots,\left<(X_\ell,x_\ell),(Y_\ell,y_\ell)\right>$ from left to right at the South pole 
if and only~if
\begin{itemize}
\item 
$X_1 = X$, $Y_\ell = Y$, $x_1 = x$, $y_\ell = y$, 
\item $C$ is left(right)-free at the North- and South-pole, respectively, if and only if $C_1$ ($C_\ell$) is, 
\item $Y_i \leq X_i$ and
\begin{itemize}
\item $C_i$ is right-free if $Y_i = X_{i+1} \in \{R^{cc},E,R^c\}$
\item $C_{i+1}$ is left-free if $Y_i = X_{i+1} \in \{L^{cc},W,L^c\}$
\item $C_i$ is right-free or $C_{i+1}$ is left-free if $Y_i \in \{L^{cc},L^c\}$ and $X_{i+1} \in \{R^{cc},R^c\}$ or vice versa.
\end{itemize}
\item 
\begin{itemize}
\item $y_i = x_{i+1} = L$ and $C_i$ is right-free, or
\item $y_i = x_{i+1} = R$ and $C_{i+1}$ is left-free, or
\item $y_i = L$ and  $x_{i+1} = R$ and $C_i$ is right-free or $C_{i+1}$ is left-free.
\end{itemize} 
\item
and single edges are the first among the components with a boundary path of type $(W,R)$ and the last among the components with a boundary path of type $(E,L)$.
\end{itemize} 
\end{restatable}

\begin{proof}[Sketch of Proof]
Since the necessity of the conditions is evident, we shortly sketch how to prove sufficiency.
By construction, we ensure that the angle between two incoming edges is 0 or $\pi$ and the angle between an incoming and an outgoing edge is $\pi/2$ or $3\pi/2$. It remains to show the following three conditions \cite{chaplick_etal:gd17}:  (i) The sum of the angles at a vertex is $2\pi$, (ii) the rotation at any inner face is $2 \pi$, (iii) and the \emph{bend-or-end property} is fulfilled, i.e., there is an assignment that assigns each edge to one of its end vertices with the following property. Let~$e_1$ and $e_2$ be two incident edges that are consecutive in the cyclic order and attached to the same side of the common end vertex $v$. Let $f$ be the face/angle between $e_1$ and $e_2$. Then at least one among $e_1$ and $e_2$ is assigned to $v$ and its bend leaves a concave angle in $f$.
\end{proof}

\cref{LEMMA:parallelComposition} yields a strict order of the possible types from left to right that can be composed in parallel.
Moreover, let $\sigma$ be a sequence of types of components from left to right that can be composed in parallel and let $\tau$ be a type in $\sigma$. 
Then \cref{LEMMA:parallelComposition} implies that $\tau$ appears exactly once in $\sigma$ or the leftmost path and the rightmost path have both the same type in $\tau$ and all four free-flags are positive. In that case the type $\tau$ might occur arbitrarily many times and all appearances are consecutive.
Thus, $\sigma$ can be expressed as a \emph{simple regular expression} on an alphabet $\mathcal T$, i.e., a sequence $\rho$ of elements in $\mathcal T \cup \{^\star\}$ such that $^\star$ does not occur as the first symbol of $\rho$  and there are no two consecutive $^\star$ in $\rho$. A sequence~$s$ of elements in $\mathcal T$ is \emph{represented} by a simple regular expression $\rho$ if it can be obtained from $\rho$ by either removing the symbol preceding a $^\star$ or by repeating it arbitrarily many times. 

Observe that the elements in the simple regular expression $\rho$ representing $\sigma$ are distinct, thus, the length of $\rho$ is linear in the number of types, i.e., constant. 
In particular, to obtain a linear-time algorithm to enumerate the attainable types of a series-parallel DAG obtained via a parallel composition, it suffices to establish the following algorithmic lemma. 

\begin{restatable}[Simple Regular Expression Matching ($\star$)]{lemma}{parallelregex}\label{LEMMA:parallel-regex}
Let $\mathcal T$ be a constant-size alphabet (set of types), and $\rho$ be a constant-length simple regular expression over $\mathcal T$.
For a collection $\mathcal{C}$ of items where each $C \in \mathcal{C}$ has a set $\mathcal T(C) \subseteq \mathcal{T}$, 
one can test in $\mathcal O(|\mathcal{C}|)$ time, if there is a selection of a type from each $\mathcal T(C)$, $C \in \mathcal{C}$ that can be ordered to obtain a sequence~represented~by~$\rho$.
\end{restatable}

\begin{corollary}The types of a parallel composition
can be computed in time linear in the number of its children.
\end{corollary}

In order to understand how the type of a series composition is determined from the types of the children, let us first have a look at an easy example: Assume that $G_1$ and $G_2$ consist both of a single edge $e_1$ and $e_2$, respectively, and that the type of both is $(W,R)$. Assume further that $G$ is obtained by merging the North poles $N_1$ and $N_2$ of $G_1$ and $G_2$, respectively. There are two ways how this can be done, namely $e_1$ can be attached to $N_1 = N_2$ before~$e_2$ or after it in the counterclockwise order starting from $R^{cc}$ and ending at $L^c$. In the first case the North type of $G$ is $R^{cc}$, in the second case it is $R^{c}$. Moreover, in the first case $G$ is left-free but not right-free at the North-pole, while in the second case it is right-free but not left-free at the South-pole. See \cref{SUBFIG:series-counterclockwise,SUBFIG:series-clockwise}.

\begin{figure}
\centering
\subcaptionbox{\label{SUBFIG:series-L}$G_1$ has type $L^{cc}$, $G_2$ type $L$}[0.45\linewidth]{
\includegraphics[page=4,scale=0.9]{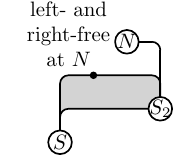}}
\subcaptionbox{\label{SUBFIG:series-LR}$G_1$ has type $L^{cc}$, $G_2$ type $R$ at $S_2$}[0.45\linewidth]{
\includegraphics[page=6,scale=0.9]{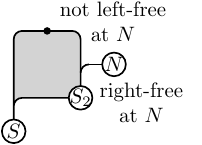} \hfil
\includegraphics[page=7,scale=0.9]{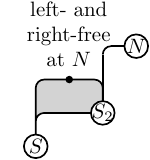}}
\caption{\label{FIG:free} Different free-flags in the case that the North pole is merged with the South pole }
\end{figure}

\begin{lemma}[Series Composition]\label{LEMMA:seriesComposition}
Let $G_1$ and $G_2$ be two series-parallel DAGs with no internal source that admit an upward-planar L-drawing of a certain type $\left<(X_1,x_1),(Y_1,y_1)\right>$ and $\left<(X_2,x_2),(Y_2,y_2)\right>$, respectively, with the poles on the outer face. 
Let $G$ be the DAG obtained by a series combination of  $G_1$ and $G_2$ such that the common pole of $G_1$ and $G_2$ is not a source in both, $G_1$ and $G_2$. Then the possible types of $G$ in an upward-planar L-drawing maintaining the types of $G_1$ and $G_2$ can be determined in constant time.
\end{lemma}
\begin{proof}
Let $S_i$ and $N_i$, respectively,  be the South and North pole of $G_i$, $i=1,2$. We may assume that $S_1$ is the South pole of $G$ and, thus, $N_1$ is the common pole of $G_1$ and $G_2$. 

First suppose that $N_1=S_2$, i.e., that $N_2$ is the North pole of $G$. It follows that $N_1$ cannot be  a source of $G_1$.
Then $G$ admits an upward-planar L-drawing if and only if $x_2 = L$ and $X_1 \neq R^{cc}$ or $y_2 = R$ and $Y_1 \neq L^{c}$, and the respective \textsc{free}-conditions are fulfilled at $N_1=S_2$.
The South-type of $G$ is the South type of $G_1$. The North-type of $G$ is the North-type of $G_2$ except for the \textsc{free}-flags, which might have to be updated if the next-to-last edge on the leftmost or rightmost path, respectively, is already contained in $G_1$ and is an outgoing edge of $N_1$. This might yield different North-types concerning the flags. See \cref{FIG:free}.

Now suppose that $N_1=N_2$. Then $G$ admits an upward-planar L-drawing if and only if $G_1$ can be embedded before $G_2$, i.e., $Y_1 \leq X_2$ and $X_1 < Y_2$, and not ($X_1 = R^{cc}$ and  $Y_2 = L^c$)  or $G_2$ can be embedded before~$G_1$, i.e., $Y_2 \leq X_1$ and $X_2 < Y_1$, and not ($X_2 = R^{cc}$ and  $Y_1 = L^c$) and the respective \textsc{free}-conditions are fulfilled at $N_1=N_2$. If $X_1 = Y_1 = X_2 = Y_2\in\{E,W\}$ then both conditions are fulfilled which might give rise to two upward-planar L-drawings with distinct labels by adding $G_2$ before or after $G_1$ at the common pole.
The \textsc{free}-flags might have to be updated if the second edge on the leftmost or rightmost path, respectively, is already contained in $G_2$ or if the next-to-last edge on the leftmost or rightmost path, respectively, is already contained in $G_1$ and the type of $G_1$ at $N_1$ equals the respective type of $G_2$ at $N_2$. Except for the flags, the South-type of $G$ is the South type of $G_1$ and the North type of $G_2$ yields the North type of $G$ except for the specifications $c$ or $cc$: First observe that both, the leftmost path and the rightmost path, either have both type $c$ or both type $cc$. Otherwise, $G_1$ would be contained in an inner face of $G_2$. The North type of both paths is indexed $c$ 
if $G_2$ is embedded before $G_1$. Otherwise, the North type is indexed $cc$. Regarding the time complexity, observe that our computation of the set of possible types of $G$ does not depend on the size of $G_1$ and $G_2$, but only on the number of types in their admissible sets. Since these sets have constant size and the above conditions on the types of $G_1$ and $G_2$ can be tested in constant time, we thus output the desired set in constant time.
\end{proof}

\begin{lemma}
The types of a series composition
can be computed in time linear in the number of its children.
\end{lemma}
\begin{proof}
Let $C_1,\dots,C_\ell$ be the components of a series component $C$ and let $\mathcal T(C_i)$, $i=1,\dots,\ell$ be the set of possible types of $C_i$. For $k=1,\dots,\ell$, we inductively compute the set $\mathcal T_i$ of possible types of the series combination $C^k$ of $C_1,\dots,C_k$, where $\mathcal T_1 = \mathcal T(C_1)$. To compute $\mathcal T_{k}$ for some $k=2,\dots,\ell$, we combine all possible combinations of a type in $\mathcal T_{k-1}$ and a type in $\mathcal T(C_k)$ and, applying \cref{LEMMA:seriesComposition}, we check in constant time which types (if any) they would yield for $C^k$. Since the  number types is constant each step can be done in constant time.
\end{proof}

\subparagraph{Single Sink.}
For the case that $G$ has a single sink, the algorithmic principles are the same as in the single-source case. The main difference is the type of an $N$-$S$-path $P$ in a component~$C$, where $S$ and $N$ are the South- and North-pole of $C$. The North pole of a component is always a sink and the North-type of $P$ is $W$ or $E$ in this order from left to right. The South-type is one among $E^{c},W^{c},L,R,E^{cc},W^{cc}$ in this order from left to right (according to the outgoing edges at $N$), depending on whether the last edge of $P$ (traversed from $N$ to $S$) is an incoming edge entering from the left (W) or the right (E), or an outgoing edge bending to the left (L) or the right (R), and whether the last edge of $P$ leaving the half-space above the horizontal line through $S$ does so to the right of $S$ (cc) or the~left~of~$S$~(c). 

The type consists again of the choice of the topmost pole (North pole), the type of the leftmost $N$-$S$-path, the type of the rightmost $N$-$S$-path, the four \textsc{free}-flags~--~which are defined the same way as in the single source case~--~and whether the component is a single edge or~not.


\section{Conclusion and Future  Work}

We have shown how to decide in linear time whether a plane single-source or -sink DAG admits an upward-planar L-drawing. A natural extension of this work would be to consider plane DAGs with multiple sinks and sources, the complexity of which is open. In the variable embedding setting, we have presented a linear-time testing algorithm for single-source or -sink series-parallel DAGs. Some next directions are to consider general single-source or -sink DAGs or general series-parallel DAGs. We remark that the complexity of testing for the existence of  upward-planar L-drawings in general also remains open.

\bibliographystyle{plainurl}
\bibliography{upl}

\clearpage 

\appendix

\section{Appendix}

\subsection{Missing proofs from \cref{SEC:preliminaries} }

\begin{figure}
\centering

\vspace{-5.5cm}

\begin{subfigure}[b]{.45\textwidth}
\begin{tikzpicture}
\node (S) {$S_\nu$};
\node (PS) [above left of=S] {};
\node (u) [left of=PS] {$u$}
edge [<-,decorate,decoration={snake}] (S);
\node (w) [above of=u] {$w$}
edge [<-] (u);
\node (PN) [below left of=w] {};
\node (N) [below left of=PN] {$N_\nu$}
edge [->,decorate,decoration={snake}] (w);
\node (outer) [left of=N] {}; 
\node (PN2) [above right of=w] {};
\node (v) [above right of=PN2] {$v$}
edge [<-,decorate,decoration={snake}] (w);
\node (dummyB) [right of=S] {};
\node (dummyBB) [below of=dummyB] {$C \in G(\mu_j)$};
\node (bottom) [above right of=dummyB] {};
\node (P) [below right of=v] {$P$};
\node (right) [right of=v] {};
\node (rright) [right of=right] {};
\node (rrright) [below right of=rright] {};
\node (rrrright) [below right of=rrright] {};
\node (rrrrright) [below right of=rrrright] {};
\node (rrrrrright) [below right of=rrrrright] {};
\node (top) [above of=v] {}; 
\node (ttop) [above of=top] {};
\node (tttop) [above of=ttop] {};
\node (ttttop) [above of=tttop] {};
\node (tttttop) [above of=ttttop] {};
\node (ttttttop) [above of=tttttop] {};
\draw [->] (S) .. controls (rrright) and (tttop)  .. (N);	
\node (NN) [below left of=N] {};
\node (NNN) [below left of=NN] {$N$}
edge [->,decorate,decoration={snake}] (N);
\node (SS) [below left of=S] {};
\node (SSS) [below left of=SS] {$S$}
edge [->,decorate,decoration={snake}] (S);
\draw [->] (SSS) .. controls (rrrrrright) and (ttttttop)  .. (NNN);	


\end{tikzpicture}
\label{SUBFIG:illustrationC}
\caption{$N_\nu$ North pole of $G(\nu)$}
\end{subfigure}
\hfill
\begin{subfigure}[b]{.45\textwidth}
\begin{tikzpicture}
\node (S) {$S_\nu$};
\node (PS) [above left of=S] {};
\node (u) [left of=PS] {$u$}
edge [<-,decorate,decoration={snake}] (S);
\node (w) [above of=u] {$w$}
edge [<-] (u);
\node (PN) [below left of=w] {};
\node (N) [below left of=PN] {$N_\nu$}
edge [->,decorate,decoration={snake}] (w);
\node (outer) [left of=N] {}; 
\node (PN2) [above right of=w] {};
\node (v) [above right of=PN2] {$v$}
edge [<-,decorate,decoration={snake}] (w);
\node (dummyB) [right of=S] {};
\node (dummyBB) [below of=dummyB] {$C \in G(\mu_j)$};
\node (bottom) [above right of=dummyB] {};
\node (right) [right of=v] {};
\node (rright) [right of=right] {};
\node (rrright) [below right of=rright] {};
\node (rrrright) [below right of=rrright] {};
\node (rrrrright) [below right of=rrrright] {};
\node (rrrrrright) [below right of=rrrrright] {};
\node (top) [above of=v] {}; 
\node (ttop) [above of=top] {};
\node (tttop) [above of=ttop] {};
\node (ttttop) [above of=tttop] {};
\node (tttttop) [above of=ttttop] {};
\node (ttttttop) [above of=tttttop] {};
\draw [->] (N) .. controls  (tttop) and (rrright)  .. (S);	
\node (NN) [below left of=N] {};
\node (NNN) [below left of=NN] {$N$}
edge [->,decorate,decoration={snake}] (N);
\node (SS) [below left of=S] {};
\node (SSS) [below left of=SS] {$S$}
edge [->,decorate,decoration={snake}] (S);
\draw [->] (SSS) .. controls (rrrrrright) and (ttttttop)  .. (NNN);	
\node (P) [below right of=v] {$P$};


\end{tikzpicture}
\label{SUBFIG:illustrationC}
\caption{$N_\nu$ South pole of $G(\nu)$}
\end{subfigure}
\caption{\label{FIG:illustration-sssp-up-nu-up}If $G'(\mu)$ contains a cycle enclosing $v$, then so does $G'(\nu)$ where $\nu$ is the grand child of $\mu$ such that $v$ is contained in $G(\nu)$. }
\end{figure}
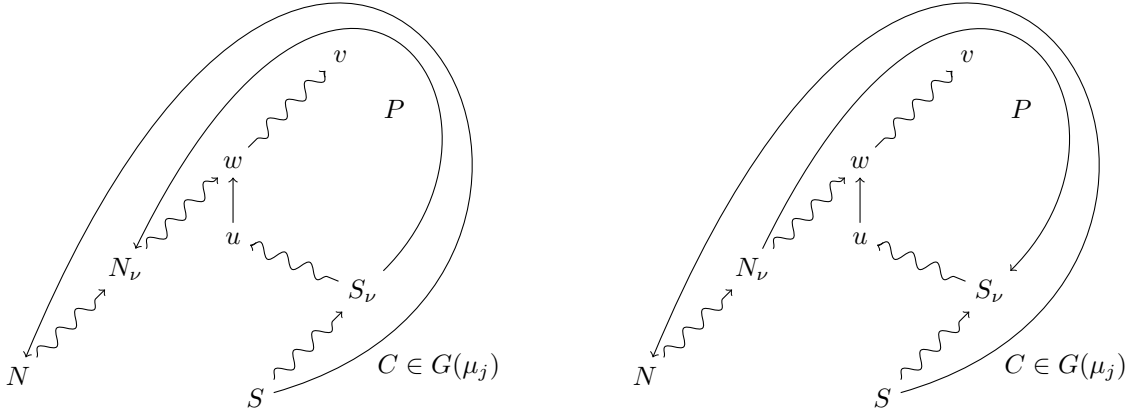

\begin{lemma}\label{LEMMA:sssp-up}
A single-source series-parallel DAG is upward-planar if and only if it is bimodal and there is no P-node $\mu$ with two children $\mu_1$ and $\mu_2$ such that the North pole~of~$G(\mu)$ 
\begin{itemize}
\item is incident to an incoming edge in $G-G(\mu)$, and
\item to both incoming and outgoing edges in both, $G(\mu_1)$ and $G(\mu_2)$.
\end{itemize}
\end{lemma}

\begin{proof} 
Clearly, a DAG has to be  bimodal in order to be upward-planar.
Let $\mu$ be  a node of the decomposition-tree of a series-parallel DAG and let $N$ be the North pole of the component $G(\mu)$.
Observe that the counterclockwise order of edges around $N$ in $G(\mu)$ in any upward-planar embedding must be as follows: some (possibly no) outgoing edges, followed by some  (possibly no) incoming edges, followed by some (possibly no) outgoing edges. Thus, if~$\mu$ has two children $\mu_1$ and $\mu_2$ such that both, $G(\mu_1)$ and $G(\mu_2)$ contain both, incoming and outgoing edges, incident to $N$ and if $\mu_1$ is before  $\mu_2$ in the counterclockwise order around $N$ then the counterclockwise order of the edges incident to $N$ is first the outgoing and then the incoming edges in $G(\mu_1)$ and first the incoming and then the outgoing edges in $G(\mu_2)$. Since an upward-planar embedding is always bimodal, it follows that in this case $N$ cannot be incident to any incoming edge outside $G(\mu)$. This implies the Necessity of the condition in the lemma. 

In order to prove the sufficiency, let $G$ be a bimodal single-source series-parallel DAG with the required property. Since a single-source DAG is upward-planar if and only if all its biconnected components are \cite{hutton/lubiw:96}, we may assume that $G$ is biconnected. Let $T$ be a decomposition-tree of $G$ with the property that the single source of $G$ is a pole of $G$.

Proceeding top-down in $T$, we first 	
check for each $P$-node $\mu$ whether it has two children~$\mu_1$ and $\mu_2$ such that the North-pole of $G(\mu)$ is incident to both incoming and outgoing edges in both, $G(\mu_1)$ and $G(\mu_2)$. If so, we arbitrarily assign one of them the \emph{marker \textsc{left}} and the other one the \emph{marker \textsc{right}}, indicating that the \textsc{left} component should be drawn before the \textsc{right} component in the counterclockwise order around $N$. All other children of a $P$-node inherit the marker of their parent. The root is assigned an arbitrary marker, say \textsc{left}. For an $S$-node $\mu$, consider a simple path $P$ in $G(\mu)$ from its South pole to its North pole. Let $\nu$ be a child of $\mu$. Then $\nu$ inherits the marker of $\mu$ if the South pole of $G(\nu)$ comes before the North pole of $G(\nu)$ on $P$. Otherwise, $\nu$ gets the marker that $\mu$ does not have.

If the North-pole of a component $G(\mu)$ is a source, let $G'(\mu)$ be the DAG that is obtained from $G(\mu)$ by adding the edge from its South- to its North-pole. Otherwise, let $G'(\mu)=G(\mu)$. In either case $G'(\mu)$ is a single source series-parallel DAG. Thus, there is always a directed path in $G'(\mu)$ from the South pole to any other vertex, in particular to the North pole.
In a bottom-up traversal of $T$, we now use the marker \textsc{right} and \textsc{left} in order to inductively construct an upward-planar embedding of $G'(\mu)$ for every node $\mu$ of $T$: 

So let $\mu$ be a node of $T$ and let $S$ and $N$ be the South and North pole of $G(\mu)$, respectively. 
If $\mu$ is a Q-node then $G(\mu)$ is a single edge and, thus, upward-planar. If $\mu$ is an S-node and $N$ is not a source then by the inductive hypothesis, the biconnected components of $G(\mu)=G'(\mu)$ are upward-planar. Thus, $G$ is also upward-planar \cite{hutton/lubiw:96}. Otherwise, $G'(\mu)$ is a parallel composition of $G(\mu)$ and the edge $(S,N)$. We handle this in the parallel case -- without using that $G(\mu)$ is upward-planar.

Now consider the case that $\mu$ is a P-node.
We call a component $G(\nu)$ with North pole $N$ an \emph{incoming component}, if all edges of $G(\nu)$ incident to $N$ are incoming edges of~$N$, \emph{outgoing}, if all edges of $G(\nu)$ incident to $N$ are outgoing edges of~$N$, and \emph{bioriented}, otherwise.
Recall that a DAG has to be bimodal in order to be upward-planar. This implies immediately that each pole can be incident to at most two bioriented components.
We define the counterclockwise order of the children $\mu_1,\dots,\mu_\ell$ around $N$ including the possibly added edge $(S,N)$ as follows: We start with a possible bioriented \textsc{left} component then all incoming components, and finally a possible bioriented \textsc{right} component. If $\mu$ was labeled \textsc{left}, we add all outgoing components at the beginning and otherwise at the end. We claim that this order yields an upward-planar embedding of $G'(\mu)$:

To this end, we use the characterization given by Hutton and Lubiw \cite{hutton/lubiw:96}: A planar embedding of a single-source DAG $G$ is upward-planar if and only if 
the source of $G$ is incident to the outer face and each vertex $v$ is a sink on the outer face of the subgraph of $G$ induced by the set Pred$(v)$ of vertices $u$ for which there is a $u$-$v$-path in $G$. 

Thus, let $v$ be a vertex of $G(\mu)$. Let Pred$_\mu(v)$ be the set of vertices $u$ for which there is a $u$-$v$-path in $G'(\mu)$, let $H_\mu(v)$ be the subgraph of $G'(\mu)$ induced by Pred$_\mu(v)$.	
Since $G$ is acyclic $v$ is a obviously a sink in $H_\mu(v)$.  Assume that $v$ is not on the outer face of $H_\mu(v)$, i.e., $H_\mu(v)$ contains a simple cycle $C$ enclosing $v$. In particular, this implies that $v$ is not a pole. Let $G(\mu_i)$, $i=1,\dots,\ell$ be the component containing $v$. 

If  $C$ is contained in $G(\mu_i)$, then $C$ is already contained in a biconnected component $G(\nu)$ of  $G(\mu_i)$ for some child $\nu$ of $\mu_i$. But then $C$ would enclose $v$ already in the subgraph of~$G'(\nu)$ induced by the set of vertices $u$ for which there is a $u$-$v$-path in $G'(\nu)$ -- contradicting that~$G'(\nu)$ is upward-planar.

So assume that $C$ contains edges from two components. In particular, this means that  both poles of $G(\mu)$ are in $C$ which implies that the North pole is incident to an outgoing edge in $G(\mu_i)$. Let $G(\mu_j)$, $i \neq j$ be a component containing edges of $C$ (including the potentially added edge between the poles). Since the South pole is not incident to an incoming edge in~$G(\mu)$ it follows that the part of $C$ in $G(\mu_j)$ is actually a directed $S$-$N$-path. In particular,~$N$ is incident to an incoming edge in $G(\mu_j)$. 

Let $1 \leq k_1 \leq k_2 \leq \ell$ be such that $G(\mu_{k_1}),\dots,G(\mu_{k_2})$ are the incoming or bioriented components of $G(\mu)$.  		
By symmetry, it suffices to consider the case $i\leq k_1$. In that case it follows that $j >i$ which implies in particular that the cycle $C$ enclosing $v$ cannot be composed by two $S$-$N$-paths in two distinct components other than $G(\mu_i)$. Thus, $C$ can only enclose $v$ if $C$ contains a not necessarily directed $S$-$N$-path in $G(\mu_i)$.

So we have the following situation: In the counterclockwise order around $N$ there are first outgoing edges from $N$-$v$-paths in $G(\mu_i)$ and then an incoming edge $e_j$ from the directed $S$-$N$-path on $C$ in $G(\mu_j)$. Let $P_N$ be the rightmost $N$-$v$-path seen from $N$ which means in particular that $P_N$ leaves $N$ first in the counterclockwise order around $N$ among all $N$-$v$-paths. 

Moreover, $\mu_i$ is labeled \textsc{left}: Either $G(\mu_i)$ is an outgoing component or the only bioriented component of $G(\mu)$ before the bioriented or incoming component $G(\mu_j)$ which implies that $\mu$ is labeled \textsc{left}. Since in that case $\mu_i$ inherits the label of $\mu$ it follows that also $\mu_i$ is labeled \textsc{left}. Otherwise, $G(\mu_i)$ is the first among the two bioriented components of $G(\mu)$. Thus, it was labeled \textsc{left}.

Observe that $v$ cannot be enclosed by $C$ if $v$ was a cut vertex of $G(\mu_i)$.
Let $\nu$ be the child of $\mu_i$ such that $v$ is contained in $G(\nu)$. 	Let $N_\nu$ be the pole of $G(\nu)$ on $P_N$ and let $S_\nu$ be the other pole of $G(\nu)$. Observe that $\nu$ is a \textsc{left} component if 
$N_\nu$ is the North-pole of~$G(\nu)$, otherwise $\nu$ is a \textsc{right} component.

Observe that both, $C$ and $P_N$ must go through $N_\nu$. 
Let $w$ be the last vertex of $C$ on~$P_N$. Then $C$ must contain an edge $(u,w)$ or $(w,u)$ in $G(\nu)$ attached to $P_N$ from the right. See \cref{FIG:illustration-sssp-up-nu-up}. Moreover, since $u$ is on $C$, there must be a $u$-$v$-path in $G$. 
Since $P_N$ was the rightmost $N$-$v$-path, there cannot be the edge there cannot be the edge $(w,u)$. But then, for the same reason there cannot be an $N$-$u$-path in $G(\mu_i)$. Thus, since $N$ and $S$ are the only possible sources of $G(\mu_i)$ there must be a directed $S$-$u$-path $P_S$ in  $G(\mu_i)$. $P_S$ and $P_N$ must be disjoint, since otherwise $G$ would contain a directed cycle or $P_N$ was not the rightmost $N$-$v$-path. The path $P_S$ must contain $S_\nu$.

If both poles of $G(\nu)$ were a source, then $N_\nu$ is the North pole of $G(\nu)$, by definition. Since $G(\nu)$ is a \textsc{left} component, we know that $G'(\nu)$ with the edge $(S_\nu,N_\nu)$ as a last edge in the counterclockwise order around $N_\nu$ is upward-planar, contradicting that $C$ together with $e$ would contain a simple cycle in  $G'(\nu)$ enclosing $v$.

Otherwise, $G(\nu)$ contains a directed path $P$ from the South pole to the North pole of $G(\nu)$. Further, observe that the North pole of $G(\nu)$ is incident to an incoming edge in $G-G(\nu)$: If~$N_\nu$ is the North pole of $G(\nu)$ then this is incoming edge of $P_N$ is $N \neq N_\nu$ and the edge on $C$ in $G(\mu_j)$ otherwise. If $S_\nu$ is the North pole of $G(\nu)$ then $S_\nu \neq S$ and its the incoming edge of $P_S$. 

Thus, if $N_\nu$ is the North pole of $G(\nu)$ then $P$ must enter $N_\nu$ after $P_N$ leaves $N_\nu$ in the counter-clockwise order around $N_\nu$. If on the other hand $S_\nu$ is the North pole of $G(\nu)$ then~$P$ must enter $S_\nu$ before $P_S$ leaves $S_\nu$ in the counter-clockwise order around $S_\nu$. 

Since $G$ is acyclic $P$ cannot contain a vertex of $P_S$ or $P_N$ other then $S_\nu$ or~$N_\nu$.
Thus, $P$, $P_N$, and $P_S$ must contain a simple cycle enclosing $v$. 
\end{proof}

\subsection{Missing proofs from \cref{SEC:fixed}} \label{APP:fixed}

\begin{lemma}\label{LEMMA:cactusPlanar}
The construction given in the proof of \cref{THEO:cacti} yields a planar L-drawing that preserves the given outerplanar embedding.
\end{lemma}
\begin{proof}

\begin{enumerate}
\item The drawing of $T$ preserves the given embedding: Each vertex has
at most one incoming edge. Let $(v,w_1)$ and $(v,w_2)$ be two outgoing edges of
a vertex $v$ such that $(v,w_1)$ is to the left of $(v,w_2)$ in the order of the outgoing edges of $v$.   Since the
children of $v$ are traversed from left to right, it follows that
$w_1$ has a lower preorder-number (y-coordinate) than $w_2$.
\item The drawing of $T$ is planar: Let $v$ be a vertex and let $w_1$ and
$w_2$ be two children of a vertex $v$ such that $w_1$ is to the
left of $w_2$. Since the children of $v$ are traversed from left
to right it follows that the vertices in subtree $T_1$ rooted at
$w_1$ obtain a lower preorder-number than the vertices in the subtree $T_2$ rooted
at $w_2$. Thus, the bounding box of $T_1$ is 
below the bounding box of $T_2$. Thus, no edge of $T_1$ can
intersect an edge of $T_2$. Further, $w_1$ and $w_2$ are the rightmost points in $T_1$ and $T_2$, 
respectively, and $v$ is to the right of both~$T_1$ and $T_2$. Thus, edges incident to $v$
cannot cross edges within $T_1$ or $T_2$. Recursively, we obtain that there are no crossings.

\item The drawing of $G$ is planar and preserves the embedding. Let now $(u,t_C)$ be the last edge on the left path of a cycle $C$. Recall that $(u,t_C)$ is not a transitive edge. Thus, by the special care we took for the backtracking, we know that the x-coordinates 
of $u$ and~$t_C$ differ only by one. This implies that $(u,t_C)$ is the leftmost incoming of $t_C$.
Moreover, $(u,t_C)$ could at most be crossed by a horizontal segment of an edge in $T$.
However, the vertices with y-coordinate between the y-coordinate of $u$ and the y-coordinate of $t_C$ are 
either in the subtree $T_u$ rooted at $u$, and thus to the left of $u$, or they are on the right 
path $P_r$ of $C$ and, thus to the right of $u$. Since there are no edges between vertices in $T_u$ 
and $P_r$, it follows that $(u,t_C)$ is not crossed. Moreover, $(u,t_C)$ is the rightmost outgoing edge of $u$. \qedhere
\end{enumerate}
\end{proof}

\begin{lemma}\label{lem:large-angles}
In a plane single-source or -sink DAG,
the number $L(f)$ of large angles in a face $f$ is $A(f)-1$, if $f$ is an inner face, and $A(f)+1$, otherwise.
\end{lemma}
\begin{proof}
By symmetry, it suffices to consider an upward-planar DAG $G$ with a single source $s$ and with a fixed upward-planar drawing. Observe that one angle at $s$ in the outer face is the only large angle of $G$ at a source-switch. We do induction on the number of biconnected components. If $G$ is biconnected this from \cite{bertolazzi_flow}. So let $C$ be a biconnected component of~$G$ that contains exactly one cut vertex $v$ and such that $C-v$ does not contain the single source~$s$. Observe that then $v$ is the single source of $C$. Let $G'$  be the DAG obtained from~$G$ by removing $C$ but not~$v$. Observe that $s$ is the single source of $G'$. By the inductive hypothesis,~$G'$ and $C$ have the required property. 
Now consider the face $f'$ of~$G'$ containing~$C$, let $f_0$ be the outer face of~$C$, and let $f$ be the face in $G$ contained in both, $f'$ and $f_0$. 

Since $v$ is a source in $C$, all sink-switches of $C$ are still sink-switches of $G$. Moreover, the upward-planar drawing of $G$ implies that the biconnected components of $G'$ that are not contained in $f_0$ cannot contain outgoing edges incident to $v$. Thus, if there were such components then the sink-switches of $G'$ are still sink-switches of $G$. Hence, in this case, we have   that $A(f) = A(f') + A(f_0)$. Moreover, the large angles at sink-switches of $G$ in $f$ are the large angles at sink-switches of $G'$ in $f'$ plus the large angles at sink-switches of $C$ in $f_0$. If $f$ is an inner face, then $f$ does not contain the large angle at the source-switch $v$ of $C$ in $f_0$. Thus $L(f) = L(f') + L(f_0) - 1 =  A(f') - 1 + A(f_0) + 1 - 1 = A(f)-1$. If $f$ is the outer face, then both, $C$ and $G'$ had a large angle at a source-switch in the outer face, but $G$ has only one. We obtain $L(f) = L(f') + L(f_0) - 1 =  A(f') + 1 + A(f_0) + 1 - 1 = A(f)+1$. Finally, if $C$ is contained in a sink-switch angle of $G'$ at $v$ then this must be a large angle and we get $A(f) = A(f') - 1 + A(f_0)$ and $L(f) = L(f') - 1 + L(f_0) - 1 = A(f') \pm 1 -1 + A(f_0) + 1 -1 = A(f) \pm 1$, where the $\pm 1$ distinguishes between the outer and the inner face.
\end{proof}

\subparagraph{Single-Sink Case of \cref{THEO:source}.} 

The results for single source DAGs translate to single sink DAGs. Each inner face $f$ has exactly one source-switch $v$ that is not large which we call bottom$(f)$. With the analogous proof as for \cref{LEMMA:path-to-top}, we obtain.

\begin{lemma}\label{LEMMA:path-to-bottom}
Let $G$ be a single sink upward-planar DAG with a fixed upward-planar embedding, let $f$ be an inner face, and let $v$ be a source with a large angle in $f$. Every plane \st-graph extending $G$ contains a directed bottom($f$)-$v$-path.
\end{lemma}

\begin{proof}[Proof of \cref{THEO:source}, single-sink case]
Let $G$ be an upward-planar single-source DAG.
Augment each inner face $f$ by edges from bottom$(f)$ to the incident source-switches. Add a new source $s$ together with edges from all source-switches on the outer face. 

Analogously as in the previous section, the thus constructed DAG has a single sink, a single source, and is acyclic. 

The edges incident to $s$ cannot be contained in a valley since their heads have all in-degree one. Also the new edges incident to bottom$(f)$ for some face $f$ cannot be the transitive edges of a valley. Finally, if an edge incident  bottom$(f)$ would be contained in a path then this path is unavoidable due to \cref{LEMMA:path-to-bottom}.

Thus, we have a planar \st-graph $G'$ that admits an upward-planar L-drawing if and only if the original DAG $G$ did. It can be tested in linear time whether $G'$ admits a bitonic \st-ordering and thus an upward-planar L-drawing~\cite{gronemann:gd16}. 
\end{proof}

\subsection{Missing Proofs from \cref{SEC:variable}}



\rotation*

\begin{proof}
Observe that $S$ is the bottommost vertex of $G$, since it is lower than $N$ and no internal vertex of $G$ is a source.
We consider the cases in which the type of $P$ at $N$ is in $\{L^c,R^c,E\}$.
Let $h$ be the horizontal line through $N$ and let $h_\ell$ ($h_r$) be the part of $h$ to the left (right) of $N$. We claim that, since $G$ has no internal source, the drawing of $P$ does not contain a subcurve $\mathcal C$ in the half plane above $h$ that connects $h_\ell$ and $h_r$: Otherwise, let~$e_\ell$ and~$e_r$, respectively, be the edges of $P$ containing the two end points of $\mathcal C$. Let $v_\ell$ and~$v_r$ be the tails of $e_\ell$ and $e_r$, respectively. Since $G$ has no internal source, there must be a descending path $P_\ell$ and $P_r$, respectively, from $v_\ell$ and $v_r$ to a pole. Since the North pole is above $v_\ell$ and $v_r$, the descending paths must end at the South pole. However, this implies that the union of $P$, $P_\ell$, and $P_r$ contains a cycle in $G$ enclosing $N$, contradicting that $N$ is incident to the outer face. 

Let $p = (x,y)$ be the end point of $P'$ different from $S$, i.e., the point below $N$. We may assume that $p$ is very close to $N$.
The above claim implies that by shortening some vertical and horizontal parts of $P'$, we can ensure that $P'$ does not traverse the horizontal line through $p$ to the left of $p$. This does not change the rotation of $P'$.
Let $x_{\min}$ be the minimum x-coordinate of $P'$. Now consider the orthogonal polyline $Q:p, (x_{\min}-1,y), (x_{\min}-1,y_S), S$ where $y_S$ is the y-coordinate of $S$. Concatenating $P'$ and $Q$ yields a simple polygon traversed in counterclockwise order. Thus the rotation of $P$ equals the rotation of a polygon minus the rotation of the two convex bends in $Q$ minus the convex bend in $S$ minus the concave bend at $p$. Thus the rotation of $P'$ is $2\pi - 3 \cdot \pi/2 + \pi/2 = \pi$. 

If the type of $P$ at $N$ is in $\{L^{cc},R^{cc},W\}$ then we obtain the respective result, by concatenating the reversion of $P'$ and the polygonal chain $S, (x_{\max}+1,y_S), (x_{\max}+1,y), p$ where $x_{\max}$ is the maximum $x$ coordinate of $P$. Thus, the reversion of $P'$ has rotation $2\pi - 3 \cdot \pi/2 + \pi/2=\pi$ and the rotation of $P'$ is the negative of it. 
\end{proof}

\parallelComposition*

\begin{proof}
\textbf{Necessity:} Assume that an upward-planar L-drawing of $C$ is given. Let $S$ and $N$ be the North- and South-pole of $C$, respectively. In an upward-planar L-drawing the edges incident to $S$ from left to right must be first all edges bending to the left and then all edges bending to the right. This yields $y_i \leq x_{i+1}$, $i=1,\dots,\ell-1$. 

Assume now that the rightmost $S$-$N$-path $P_R$ of $C_i$ has South-type $L$ and that the bend~$b$ of the first edge $e$ of $P_R$ is contained in the horizontal segment of another edge $e'$ of $C_i$. Then the first edge $e''$ of the leftmost path of $C_{i+1}$ cannot bend to the left. Otherwise, $e''$ would have to intersect $e$ or $e'$. Analogously, $e$ cannot bend to the right if the leftmost path of~$C_{i+1}$ has South-type $R$ and the bend on $e''$ is contained in another edge of $C_{i+1}$. Finally consider the case that $e$ bends to the left and $e''$ bends to the right. If the vertical segment of $e$ is longer than the vertical segment of $e''$, then the bend of $e''$ cannot be contained in the horizontal segment of another edge and vice versa. 

Any set of disjoint $S$-$N$-paths in the order from left to right at the South pole must have the North type in this order: $R^{cc},L^{cc},W,E,R^c,L^c$ where there cannot be both, paths of type $L^c$ and paths of type $R^{cc}$.
It follows that $Y_i \leq X_{i+1}$, $i=1,\dots,\ell-1$. 
If 	the leftmost $S$-$N$-path $P_L$ of a component $C_i$ is of type $R^{cc}$ ($E$) and if the bend $b$ in the last edge $e$ of~$P_L$ is not free at the North-pole, i.e., if there is an edge $e'$ in $C_i$ that is not incident to $N$ but contains $b$, then the rightmost $S$-$N$-path of $C_{i+1}$ cannot be of type $R^{cc}$ ($E$). 
Similarly, if 	the rightmost $S$-$N$-path $P_R$ of a component $C_{i+1}$ is of type $L^{c}$ ($W$) and if the bend $b$ in the last edge $e$ of $P_R$ is not free at the North-pole, i.e., if there is an edge $e'$ in $C_{i+1}$ that is not incident to $N$ but contains $b$, then the leftmost $S$-$N$-path of $C_i$ cannot be of type~$L^{c}$~($W$). 
Moreover, since $C$ has no source other than the poles, the rightmost path of type~$R^{c}$ and a leftmost path of type $L^{cc}$ is always free.
Finally, if the last edges of $P_R$ and $P_L$ bend to opposite sides, then the bend on the edge with the shorter vertical segment cannot be contained in the horizontal segment of another edge. 

Assume now that there is an edge $(S,N)$ of North type $W$. Let $C_i$ be the component preceding $(S,N)$ and assume that the North-type of the rightmost path $P_R$ of $C_i$ is of type~$W$. Then the South type of $P_R$ cannot be $R$. Similarly, if $(S,N)$ has North type $E$ and the leftmost path $P_L$ of the component succeeding  $(S,N)$ has also North type $E$, then the South type of $P_L$ must be $R$.

\textbf{Sufficiency:} 
By construction, we ensure that the angle between two incoming edges is 0 or $\pi$ and the angle between an incoming and an outgoing edge is $\pi/2$ or $3\pi/2$. It remains to show the following three conditions \cite{chaplick_etal:gd17}:  (i) The sum of the angles at a vertex is $2\pi$, (ii) the rotation at any inner face is $2 \pi$, (iii) and the \emph{bend-or-end property} is fulfilled, i.e., there is an assignment that assigns each edge to one of its end vertices with the following property. Let~$e_1$ and $e_2$ be two incident edges that are consecutive in the cyclic order and attached to the same side of the common end vertex $v$. Let $f$ be the face/angle between $e_1$ and $e_2$. Then at least one among $e_1$ and $e_2$ is assigned to $v$ and its bend leaves a concave angle in $f$.

We first show that the bend-or-end property is fulfilled. To this end, we have to define an assignment of edges to end vertices:
Consider upward-planar L-drawings of the components respecting the given types. Consider an edge $e$ whose bend is contained in another edge $e'$ of the same component.  Let $v$ be the common end vertex of $e$ and $e'$. Assign the bend on $e$ to $v$. 
Now consider two consecutive components $C_i$ and $C_{i+1}$, let $P_R$ be the rightmost $S$-$N$-path of $C_i$ and let $P_L$ be the leftmost $S$-$N$-path of $C_{i+1}$. If $P_R$ and $P_L$ both have South-type $L$, then $C_i$ must be right-free. Thus, we can assign the first edge of $P_R$ to $S$ without violating any previous assignments.
Similarly, if $P_R$ and $P_L$ both have South-type $R$, then we can assign the first edge of $P_L$ to $S$ without violating any previous assignments. If $P_R$ has South type $L$ and $P_L$ has South type $R$, then we assign the first edge of $P_R$ or $P_L$ to $S$ depending on whether $C_i$ is right-free or $P_L$ is left-free at the South pole.

Consider now the last edge $e_R$ and $e_L$ of $P_R$ and $P_L$, respectively. If $e_R$ and $e_L$ are both incoming edges of $N$, the assignment is analogous to the South pole: If $P_R$ and $P_L$ have both North-type $L^c$ or both $L^{cc}$, we assign $e_L$ to $N$. If  $P_R$ and $P_L$ have both North-type $R^c$ or both $R^{cc}$, we assign $e_R$ to $N$. If $P_R$ has North type  $L^c$ or $L^{cc}$ and $P_L$ has North type $R^c$ or $R^{cc}$, or vice versa, then we assign $e_R$ or $e_L$ to $N$ depending on whether $C_i$ is right-free or~$C_{i+1}$ is left-free at the North pole. 

Consider now the case that $P_R$ and $P_L$ are both of type $W$. Then $C_{i+1}$ is left-free. Moreover, if $C_{i+1}$ is a single edge, then $C_{i}$ cannot have South-type $R$, thus, we have not assigned $e_L$ to $S$. Hence, we can assign $e_L$ to $N$. Finally, if  $P_R$ and $P_L$ are both of type $E$ we can assign $e_R$ to $N$.  We assign all edges that have not been assigned yet to an arbitrary end vertex. This assignment fulfills the bend-or-end property.

By the construction, 
the angular sum around a vertex is always $2 \pi$. It remains to show that the rotation of all inner faces is $2 \pi$, which implies that the rotation of the outer face is~$-2 \pi$. If $f$ is an inner face of a component, then its rotation is $2 \pi$ by induction. So consider the face $f_i$ between the components $C_i$ and $C_{i+1}$. The boundary of $f_i$ in counter-clockwise direction is combined by the reversed rightmost path $P_R$ of $C_i$ and the leftmost path $P_L$ of~$C_{i+1}$.
%
Let $P_R'$ be the concatenation of $P_R$ and a short vertical segment $s$ emanating under~$N$ and let $P_L'$ be the concatenation of $P_L$ and $s$. Consider the face $f'_i$ bounded by $P_R'$ and the reversion of $P_L'$. Then $f_i$ has the same rotation as $f'_i$. Let rot$(P)$ be the rotation of a path. The rotation of $f'_i$ at $S$ is $2\cdot \pi/2$. The rotation of $f'_i$ at the bottommost end point~$p$ of $s$ is $2 \cdot (-\pi/2)$ or $2 \cdot \pi/2$, depending on whether the $2 \pi$ angle between $P_R'$ and $P_L'$ at~$p$ is in the interior of $f_i'$ or not, i.e.\ whether $P_L$ and $P_R$ pass to different sides of $N$ or not. Thus, \mbox{rot$(f_i) = $ rot$(P_R') \pm \pi  -$ rot$(P_L') + \pi$}. Since rot$(P_R')=-$rot$(P_L')$ if and only if $P_L$ and~$P_R$ pass to different sides of $N$, it follows that the rotation of $f_i$ is always $2 \pi$.
\end{proof}

\parallelregex*

\begin{proof}
%
Since we are allowed to freely order our selection of values (one for each element of~$\mathcal{C}$), the order of entries in $\rho$ does not matter.
For each $\tau \in \mathcal{T}$, let $\aleph(\tau)$ denote the number of times $\tau$ occurs in $\rho$  without the $\star$ operator; note, $\sum_{\tau \in \mathcal{T}} \aleph(\tau) \leq |\rho|$ which is constant.
Using this notation, the task at hand becomes a kind of matching problem where we simply need to make a selection of values so that, for each type $\tau$ in $\mathcal{T}$, either:
\begin{itemize}
\item $\tau$ must be selected at least $\aleph(\tau)$ times and is called \emph{unbounded}, if $\tau^\star$ occurs in $\rho$; or  
\item $\tau$ must be selected exactly $\aleph(\tau)$ times and is called \emph{exact}, otherwise.
\end{itemize}

We now provide an $\mathcal O(|\mathcal{C}|)$ time algorithm to solve this problem.
First, we partition $\mathcal{T}$ into two parts based on the number of items in $\mathcal{C}$ that can be selected with a given value: 
\begin{itemize}
\item $\mathcal{T}_{few} = \{ \tau \in \mathcal{T} : $ at most $2|\rho|$ items $C$ in $\mathcal{C}$ have $\tau \in \mathcal T(C)$\}, and 
\item $\mathcal{T}_{many} = \{ \tau \in \mathcal{T} : $ more than $2|\rho|$ items $C$ in $\mathcal{C}$ have $\tau \in \mathcal T(C)$\}.
\end{itemize}
Observe that the number of items in $\mathcal{C}$ that attain the types in $\mathcal{T}_{few}$ is at most $|\mathcal{T}| \cdot 2|\rho|$; namely, this is only constantly many. 
Let $\mathcal{C}_{few}$ be this subset, and let $\mathcal{C}_{many}$ be $\mathcal{C} \setminus \mathcal{C}_{few}$. 
Now, we enumerate all possible selections of types for the set~$\mathcal{C}_{few}$ as to precisely satisfy the~$\aleph(\tau)$ requirement for each $\tau \in \mathcal{T}_{few}$.  If no such a selection exists, then we conclude that the test is negative. Otherwise, 
let $\alpha$ be any such a partial selection. 
This may result in some elements of $\mathcal{C}_{few}$ with no selected type so far; let $\mathcal{C}_{few}'$ be this unassigned subset of $\mathcal{C}_{few}$. 

At this point we still need to meet the $\aleph(\tau)$  requirement for each $\tau \in \mathcal{T}_{many}$ (without further selecting any exact $\tau$ from $\mathcal{C}_{few}$). 
We will do this in three steps. 
First, for each $C \in \mathcal{C}_{many} \cup \mathcal{C}_{few}'$, if there is an unbounded type in $\mathcal T(C)$, we select an arbitrary one for~$C$ (this may be changed later). 
For the second step, let $\mathcal{C}^*$ be the subset of $\mathcal{C}$ for which no type has been selected. 
Note that, for each $C \in \mathcal{C}^*$, the set $\mathcal T(C)$ only contains exact types.
Thus, $|\mathcal{C}^*| \leq |\rho|$, otherwise we may reject the partial selection $\alpha$, and go to the next one.
We now check if there is a selection of types for $\mathcal{C}^*$ that does not exceed any $\aleph(\tau)$, this can be done in constant time and if no such selection exists we can safely reject the partial selection $\alpha$, and go to the next one.
Finally, in the third step, we have now selected a type for every item in $\mathcal{C}$. 
Moreover, in this selection, due to how we have chosen $\mathcal{T}_{many}$, we can now greedily reallocate (as needed) the items assigned to unbounded types in the first step as to fulfill all of the $\aleph(\tau)$ requirements for all types. 
\end{proof}

\end{document}